\documentclass[article,nofootinbib]{revtex4}
\usepackage{nicefrac}
\usepackage[english]{babel}

\usepackage{hyperref}
\usepackage{dcolumn}
\usepackage{xcolor}
\definecolor{mylinkcolor}{rgb}{0,0,0.4}
\hypersetup{unicode=true, %
  bookmarksnumbered=false,bookmarksopen=false,bookmarksopenlevel=1, %
  breaklinks=true,pdfborder={0 0 0},colorlinks=true}%
\hypersetup{%
  anchorcolor=mylinkcolor,citecolor=mylinkcolor, %
  filecolor=mylinkcolor,linkcolor=mylinkcolor, %
  menucolor=mylinkcolor,runcolor=mylinkcolor, %
  urlcolor=mylinkcolor}%

\usepackage{bbm}
\usepackage{amsmath}
\usepackage{graphicx}
\usepackage{amsthm}
\usepackage{amssymb}
\usepackage[makeroom]{cancel}
\usepackage{mathrsfs}
\usepackage{MnSymbol}
\usepackage{braket}
\usepackage{epstopdf}
\usepackage{tikz}
\usetikzlibrary{arrows}

\newcommand{\indep}{\rotatebox[origin=c]{90}{$\models$}}

\newcommand{\cC}{{\rm C}}
\newcommand{\qQ}{{\rm Q}}
\newcommand{\gG}{{\rm G}}
\newcommand{\cE}{\mathcal{E}}
\newcommand{\cF}{\mathcal{F}}
\newcommand{\cH}{\mathcal{H}}
\newcommand{\cI}{\mathcal{I}}
\newcommand{\cM}{\mathcal{M}}
\newcommand{\cP}{\mathcal{P}}
\newcommand{\cS}{\mathcal{S}}
\newcommand{\iI}{{\rm I}}

\newcommand{\comment}[1]{}

\newtheorem{thm}{\bf Theorem}
\newtheorem{prop}[thm]{\bf Proposition}
\newtheorem{lemma}[thm]{\bf Lemma}

\newtheorem*{thm*}{Theorem}
\newtheorem*{prop*}{Proposition}
\newtheorem*{lemma*}{Lemma}
\newtheorem*{cor*}{Corollary}
\newtheorem*{conj*}{Conjecture}
\newtheorem*{idea*}{Idea}
\newtheorem*{remark*}{Remark}
\theoremstyle{definition}
\newtheorem{example}{Example}
\newtheorem{definition}[thm]{Definition}

\begin{document}

\title{Analysing causal structures with entropy}

\date{$29^{\text{th}}$ November 2017}

\author{Mirjam Weilenmann}
\email{msw518@york.ac.uk}
\address{Department of Mathematics, University of York,
  Heslington, York, YO10 5DD, UK.}
\author{Roger Colbeck}
\email{roger.colbeck@york.ac.uk}
\address{Department of Mathematics, University of York,
  Heslington, York, YO10 5DD, UK.}

\begin{abstract} 
  A central question for causal inference is to decide whether a set
  of correlations fit a given causal structure. In general, this
  decision problem is computationally infeasible and hence several
  approaches have emerged that look for certificates of
  compatibility. Here we review several such approaches based on
  entropy. We bring together the key aspects of these entropic
  techniques with unified terminology, filling several gaps and
  establishing new connections regarding their relation, all
  illustrated with examples.  We consider cases where unobserved
  causes are classical, quantum and post-quantum and discuss what
  entropic analyses tell us about the difference. This has
  applications to quantum cryptography, where it can be crucial to
  eliminate the possibility of classical causes. We discuss the
  achievements and limitations of the entropic approach in comparison
  to other techniques and point out the main open problems.
\end{abstract}

\maketitle

\section{Introduction}
Deciding whether a causal explanation is compatible with given
statistical data and exploring whether it is the most suitable
explanation for the data at hand are central scientific tasks.
Sometimes the most reasonable explanation of a set of observations
involves unobserved common causes.  In the case where the common
causes are classical, the well-developed machinery of Bayesian
networks can be used~\cite{Pearl2009, Spirtes2000}. In principle, such
networks are well-understood and it is known how to check whether
observed correlations are compatible with a given
network~\cite{Geiger1999}.  In practice, however, testing
compatibility for networks that involve unobserved systems is only
computationally tractable for small cases~\cite{Garcia2005,Lee2015}.
Furthermore, the methodology has to be adapted whenever non-classical
common causes are permitted.

Finding good heuristics to help identify correlations that are
(in)compatible with a causal structure is currently an active area of
research~\cite{Chaves2012, Chaves2013, Fritz2013, Chaves2014,
  Chaves2014b, Henson2014, Steudel2015, Chaves2015, Rosset2015,
  Chaves2015a, Pienaar2015, Chaves2016, Pienaar2016, Wolfe2016,
  Kela2017} and the use of entropy measures is common to many of
these~\cite{Chaves2012, Chaves2013, Fritz2013, Chaves2014,
  Chaves2014b, Henson2014, Steudel2015, Chaves2015, Chaves2016,
  Pienaar2016, Kela2017}.  Such methods are important in the quantum
context, where recent cryptographic protocols rely on the lack of a
classical causal explanation for certain quantum correlations in
specified causal
structures~\cite{Ekert1991,Mayers1998,Barrett2005b,Acin2006,Colbeck2009,
  Colbeck2011,Pironio2010,Vazirani2014,Miller2014}, an idea that lies
behind Bell's theorem~\cite{Bell1964} (see also~\cite{Wood2012}).

In Section~\ref{sec:entropicappr} of this article, we review the
entropic characterisation of the correlations compatible with causal
structures in classical, quantum and more general non-signalling
theories. We detail refinements of the approach based on
post-selection in Section~\ref{sec:post-selection}.  Together, these
sections show the current capabilities of entropic techniques, also
establishing and clarifying connections between different
contributions. Our review is illustrated with several examples to
assist its understanding and to make it easily accessible for
applications. In Section~\ref{sec:further_techniques} we outline and
compare further approaches to the problem, before concluding in
Section~\ref{sec:conclusion} with some open questions.

\section{Entropy vector approach}\label{sec:entropicappr}
Characterising the joint distributions of a set of random variables or
alternatively considering a multi-party quantum state in terms of its
entropy (and of those of its marginals) has a tradition in information
theory, dating back to Shannon~\cite{Shannon1948, Yeung1997,
  Pippenger2003}.  However, only recently has this approach been
extended to account for causal structure~\cite{Chaves2012, Fritz2013}.
In Sections~\ref{sec:classicalcone} and~\ref{sec:causal} respectively,
we review this approach with and without imposing causal
constraints. All our considerations are concerned with discrete random
variables, for extensions of the approach to continuous random
variables (and its limitations) we refer to~\cite{Chan2003,
  Fritz2013}.

\subsection{Classical entropy cone} \label{sec:classicalcone}
The entropy cone for a joint distribution of $n$ random variables was
introduced in~\cite{Yeung1997}. It is defined in terms of the
\emph{Shannon entropy}~\cite{Shannon1948}, which for a discrete random
variable $X$ taking values $x \in \mathcal{X}$ with probability
distribution $P_X$ is defined by
$$H(X) := - \sum_{x \in \mathcal{X}} P_X(x) \log_\mathrm{2}
P_X(x)\, ,$$
where $0 \log_\mathrm{2} (0)$ is taken to be $0$\footnote{Note 
$\lim_{p \rightarrow 0^{+}} p \log_\mathrm{2}p=0$}.

For a set of $n\geq 2$ jointly distributed random variables,
$\Omega:=\left\{X_{1},~X_{2},~\ldots~,~X_{n}\right\}$, we denote their
probability distribution as $P \in \mathcal{P}_{n}$, where $\cP_{n}$
is the set of all probability mass functions of $n$ jointly
distributed random variables. The Shannon entropy maps any subset of
$\Omega$ to a non-negative real value:
${ H: \mathscr{P}(\Omega) \rightarrow \left[0, \infty \right), \
X_S \mapsto H(X_S),}$
where $\mathscr{P}(\Omega)$ denotes the power set of $\Omega$, and $H(\{\})=0$.
The entropy of the joint distribution of the random variables $\Omega$
and of all its marginals can be expressed as components of
a vector in $\mathbbm{R}_{\geq 0}^{2^n-1}$, ordered in the following
as\footnote{Since the empty set always has zero entropy, we choose to omit it
from the entropy vector in this work.}
$$(H(X_1),~H(X_2),~\ldots~,~H(X_n), H(X_1X_2),~H(X_1X_3),~\ldots~,~H(X_1 X_2~\ldots~X_n)).$$
We use
${\bf H}(P) \in \mathbb{R}^{2^{n}-1}_{\geq 0}$ to denote the vector
corresponding to a particular distribution $P \in \mathcal{P}_{n}$.
The set of all such vectors is
$\Gamma^*_n:= \left\{v \in \mathbbm{R}_{\geq 0}^{2^{n}-1} \mid \exists
  P \in \mathcal{P}_n \text{ s.t. } v={\bf{H}}(P) \right\}.$ Its
closure $\overline{\Gamma^*_n}$ includes vectors $v$ for which there
exists a sequence of distributions $P_k\in\cP_n$ such that
${\bf H}(P_k)$ tends to $v$ as $k\rightarrow \infty$.  It is known
that the \emph{entropy cone} $\overline{\Gamma^*_n}$ is a convex cone
for any $n\in \mathbb{N}$~\cite{Zhang1997}.  As such, its boundary may
be characterised in terms of (potentially infinitely many) linear
inequalities. Because $\overline{\Gamma^*_n}$ is difficult to
characterise, we will in the following consider various
approximations.

\subsubsection{Outer approximation: the Shannon cone}
The standard outer approximation to $\overline{\Gamma^*_n}$ is the
polyhedral cone constrained by the \emph{Shannon inequalities} listed
in the following\footnote{It is a matter of convention, whether
  $H(\{\})=0$ is included as a Shannon inequality; we keep this
  implicit.}:
\begin{itemize}
\item Monotonicity: For all $X_T$, $X_S \subseteq \Omega$,
${H(X_S\setminus X_T) \leq H(X_S)}$.
\item Submodularity: For all $X_S$, $X_T \subseteq \Omega$,
${H(X_S \cap X_T) +  H(X_{S} \cup X_{T})  \leq  H(X_{S}) + H(X_{T})}$.
\end{itemize}

These inequalities are always obeyed by the entropies of a set of
jointly distributed random variables. They may be concisely rewritten
in terms of the following information measures: the \emph{conditional
  entropy} of two jointly distributed random variables $X$ and $Y$,
${H(X|Y):=H(XY)-H(Y)}$, their \emph{mutual information},
${I(X:Y):=H(X)+H(Y)-H(XY)}$, and the \emph{conditional mutual
  information} between two jointly distributed random variables $X$
and $Y$ given a third, $Z$, denoted
${I(X:Y|Z):=H(XZ)+H(YZ)-H(Z)-H(XYZ)}$. Hence, the monotonicity
constraints correspond to positivity of conditional entropy,
$H(X_S\cap X_T|X_S\setminus X_T)\geq 0$, and submodularity is
equivalent to positivity of the conditional mutual information,
${ I(X_S \setminus X_T :X_T \setminus X_S |X_S \cap X_T) \geq 0 }$.
The monotonicity and submodularity constraints can all be generated
from a minimal set of $n + n (n-1) 2^{n-3}$
inequalities~\cite{Yeung1997}: for the monotonicity constraints it is
sufficient to consider the $n$ constraints with $X_S=\Omega$ and
$X_T=X_i$ for some $X_i\in\Omega$; for the submodularity constraints
it is sufficient to consider those with $X_S \setminus X_T=X_i$ and
$X_T \setminus X_S=X_j$ with $i<j$ and where $X_U:=X_S \cap X_T$ is
any subset of $\Omega$ not containing $X_i$ or $X_j$, i.e.,
submodularity constraints of the form ${ I(X_i :X_j |X_U) \geq 0 }$.

These $n+n(n-1)2^{n-3}$ independent Shannon inequalities can be
expressed in terms of a ${(n+n(n-1)2^{n-3})\times(2^{n}-1)}$
dimensional matrix, which we call $M^{n}_\mathrm{SH}$, such that for
any $v\in\Gamma^{*}_n$, the conditions
$M^{n}_\mathrm{SH}\cdot v\geq 0$ hold\footnote{This condition is to be
  interpreted as the requirement that each component of
  $M^{n}_\mathrm{SH}\cdot v$ is non-negative.}.  More generally, for
$v\in\mathbb{R}_{\geq0}^{2^{n}-1}$, a violation of
$M^{n}_\mathrm{SH}\cdot v\geq 0$ certifies that there is no
distribution $P\in\cP_n$ such that $v={\bf H}(P)$. It follows that the
\emph{Shannon cone},
${\Gamma_n:=\left\{v\in\mathbbm{R}_{\geq 0}^{2^{n}-1}\mid
    M^n_{\mathrm{SH}}\cdot v\geq 0\right\}}$,
is an outer approximation of the set of achievable entropy vectors,
$\Gamma_n^*$~\cite{Yeung1997}.

\begin{example}\label{example:3shannon}
The three variable Shannon cone is 
$\Gamma_3=\left\{v\in\mathbbm{R}_{\geq 0}^7\mid M^3_{\mathrm{SH}}\cdot v\geq 0\right\}$, where
\begin{equation*}
M^{3}_\mathrm{SH}= \left( \begin{array}{ccccccc}
0 & 0 & 0 & 0 & 0 & -1 & 1 \\
0 & 0 & 0 & 0 & -1 & 0 & 1 \\
0 & 0 & 0 & -1 & 0 & 0 & 1 \\
1 & 1 & 0 & -1 & 0 & 0 & 0 \\
1 & 0 & 1 & 0 & -1 & 0 & 0 \\
0 & 1 & 1 & 0 & 0 & -1 & 0 \\
-1 & 0 & 0 & 1 & 1 & 0 & -1 \\
0 & -1 & 0 & 1 & 0 & 1 & -1 \\
0 & 0 & -1 & 0 & 1 & 1 & -1 
\end{array} \right) .
\end{equation*}
The first three rows are monotonicity constraints, the remaining six ensure submodularity.
\end{example}

\subsubsection{Beyond the Shannon cone}\label{sec:nonShan}
For two variables the Shannon cone coincides with the actual entropy
cone, $\Gamma_{2}= \Gamma_{2}^{*}$, while for three random variables
this holds only for the closure of the entropic cone, i.e.\
$\Gamma_{3}= \overline{\Gamma_{3}^*}$ but
$\Gamma_{3}\neq\Gamma_{3}^{*}$~\cite{Han81,Zhang1997}. For $n \geq 4$
further independent constraints on the set of entropy vectors are
needed to fully characterise $\overline{\Gamma_{n}^*}$, the first of
which was discovered in~\cite{Zhang1998}.

\begin{prop}[Zhang \& Yeung] \label{prop:zhangyeung} For any four discrete random variables $T$, $U$, $V$ and $W$ the following inequality holds:
$- H(T) - H(U) -\frac{1}{2} H(V)  + \frac{3}{2}H(TU) + \frac{3}{2}H(TV) 
+ \frac{1}{2}H(TW) + \frac{3}{2}H(UV)  + \frac{1}{2}H(UW) 
 -\frac{1}{2}H(VW) - 2 H(TUV) - \frac{1}{2}H(TUW) \geq 0$.
\end{prop}

For $n\geq4$ the convex cone
$\overline{\Gamma_{n}^*}\subsetneq\Gamma_{n}$ is not polyhedral, i.e.,
it cannot be characterised by finitely many linear
inequalities~\cite{Matus2007}. Nonetheless, many linear entropic
inequalities have been discovered~\cite{Zhang1997, Zhang1998,
  Makarychev2002, Dougherty2006}. Recently, systematic searches for
new entropic inequalities for $n=4$ have been conducted~\cite{Xu2008,
  Dougherty2011}, which recover most of the previously known
inequalities; in particular the inequality of
Proposition~\ref{prop:zhangyeung} is re-derived and shown to be
implied by tighter ones~\cite{Dougherty2011}. The systematic search
in~\cite{Dougherty2011} is based on considering additional random
variables that obey certain constraints and then deriving four
variable inequalities from the known identities for five or more
random variables (see also~\cite{Zhang1998, Matus2007}), an idea that
is captured by a so-called copy lemma~\cite{Zhang1998, Dougherty2011,
  Kaced2013}.  In the same article, several rules to generate families
of inequalities have been suggested, in the style of techniques
introduced by Mat\'{u}\v{s}~\cite{Matus2007}.

For more than four variables, a few additional inequalities are
known~\cite{Zhang1998, Makarychev2002, Zhang2003}.  Curiously, to our
knowledge, in the case of four variables, all known relevant non
Shannon inequalities (i.e., the ones found in~\cite{Zhang1998,
  Makarychev2002, Dougherty2006, Matus2007, Xu2008, Dougherty2011}
that are not yet superseded by tighter ones) can be written as a
positive linear combination of the \emph{Ingleton quantity},
${I(T:U|V) + I(T:U|W) + I(V:W) - I(T:U)}$, and conditional mutual
information terms (see also~\cite{Dougherty2011}).

\subsubsection{Inner approximations}
For the four variable entropy cone, $\overline{\Gamma_\mathrm{4}^*}$,
an inner approximation is defined as the region constrained by the
Shannon inequalities and the six permutations of the \emph{Ingleton
  inequality}~\cite{Ingleton},
${I(T:U|V) + I(T:U|W) + I(V:W) - I(T:U)\geq 0}$, for random variables
$T$, $U$, $V$ and $W$. These inequalities can be concisely written as
a matrix $M_\mathrm{I} \in \mathbb{R}^{6}\times \mathbb{R}^{15}$. The
constrained region is called the \emph{Ingleton cone},
$\Gamma^{\iI}:= \left\{v \in \mathbbm{R}_{\geq 0}^{15} \mid
  M^{4}_\mathrm{SH} \cdot v \geq 0 \text{ and } M_\mathrm{I} \cdot v
  \geq 0 \right\},$ and it has the property that $v \in \Gamma^{\iI}$
implies $v \in \Gamma^*_n$~\cite{Hammer2000}. In contrast, there are
entropy vectors that violate the Ingleton inequalities, as the
following example shows.

\begin{example}
  Let $T$, $U$, $V$ and $W$ be four jointly distributed random
  variables. Let $V$ and $W$ be uniform random bits and let
  $T=\operatorname{AND}(\neg V, \neg W)$ and
  $U=\operatorname{AND}(V,W)$.  This distribution~\cite{Matus2007}
  leads to the entropy vector
  ${ v \! \approx \!
    (0.81,0.81,1,1,1.50,1.50,1.50,1.50,1.50,2,2,2,2,2,2),}$ for which
  ${ I(T:U|V) + I(T:U|W) + I(V:W) - I(T:U) \approx -0.12 }$ in
  violation of the Ingleton inequality.
\end{example}
For five random variables an inner approximation in terms of Shannon,
Ingleton and $24$ additional inequalities and their permutations is
known (including partial extensions to more
variables)~\cite{Dougherty2009, Dougherty2014}.

\subsection{Entropy vectors for causal structures} \label{sec:causal}
Causal relations among a set of variables impose constraints on their
possible joint distributions, which can be conveniently represented
with a causal structure.

\begin{definition}
  A \emph{causal structure} is a set of variables arranged in a
  directed acyclic graph~(DAG), in which a subset of the nodes is
  assigned as observed.
\end{definition}

The directed edges of the graph are intended to represent causation,
perhaps by propagation of some influence, and cycles are excluded to
avoid the well-known paradoxes associated with causal loops.  We will
interpret causal structures in different ways depending on the
supposed physics of whatever is mediating the causal influence.

One of the simplest causal structures that leads to interesting
insights and one of the most thoroughly analysed ones is Pearl's
instrumental causal structure, $IC$~\cite{Pearl1995}. It is displayed
in Figure~\ref{fig:instrumental}(a) and will be used as an example
throughout this review.
\begin{figure}
\centering 
\resizebox{0.85\columnwidth}{!}{
\begin{tikzpicture}
\node (A) at (-2.5,1.5) {$(a)$};
\node[draw=black,circle,scale=0.75] (1) at (-2,-0.5) {$X$};
\node[draw=black,circle,scale=0.75] (2) at (-0,-0.5) {$Z$};
\node[draw=black,circle,scale=0.75] (3) at (2,-0.5) {$Y$};
\node (4) at (1,1) {$A$};

\draw [->,>=stealth] (1)--(2);
\draw [->,>=stealth] (2)--(3);
\draw [->,>=stealth] (4)--(2) node [above,pos=0.8,yshift=+1ex] {$\scriptstyle A_\mathrm{Z}$};
\draw [->,>=stealth] (4)--(3) node [above,pos=0.8,yshift=+1ex] {$\scriptstyle A_\mathrm{Y}$};

\node (B) at (3.5,1.5) {$(b)$};
\node[draw=black,circle,scale=0.75] (B1) at (4,-0.5) {$A$};
\node(B2) at (6,-0.5) {$C$};
\node[draw=black,circle,scale=0.75] (B3) at (8,-0.5) {$B$};
\node[draw=black,circle,scale=0.75] (B4) at (5,1) {$X$};
\node[draw=black,circle,scale=0.75] (B5) at (7,1) {$Y$};

\draw [->,>=stealth] (B1)--(B4);
\draw [->,>=stealth] (B2)--(B4) node [below,pos=0.8,yshift=-1.5ex] {$\scriptstyle C_\mathrm{X}$};
\draw [->,>=stealth] (B2)--(B5) node [below,pos=0.8,yshift=-1.5ex] {$\scriptstyle C_\mathrm{Y}$};
\draw [->,>=stealth] (B3)--(B5);  

\node (C0) at (9.5,1.5) {$(c)$};
\node[draw=black,circle,scale=0.75] (X) at (10.5,1) {$X$};
\node[draw=black,circle,scale=0.75] (Y) at (12.5,1) {$Y$};
\node[draw=black,circle,scale=0.75] (Z) at (11.5,-0.5) {$Z$};
\node (A) at (12,0.28) {$A$};
\node (B) at (11,0.28) {$B$};
\node (C) at (11.5,1) {$C$};

\draw [->,>=stealth] (A)--(Y);
\draw [->,>=stealth] (A)--(Z);
\draw [->,>=stealth] (B)--(X);
\draw [->,>=stealth] (B)--(Z);
\draw [->,>=stealth] (C)--(X);
\draw [->,>=stealth] (C)--(Y);
\end{tikzpicture}
}
\caption{(a) Pearl's instrumental scenario. The nodes $X$, $Y$ and $Z$
  are observed, $A$ is unobserved.  In the classical case this can be
  understood in the following way: A random variable $X$ and an
  unobserved $A$ are used to generate another random variable
  $Z$. Then $Y$ is generated from $A$ and the observed output of node
  $Z$. In particular, note that no other information can be forwarded
  from $X$ through the node $Z$ to $Y$.  In the quantum case, the
  source $A$ shares a quantum system $\rho_A\in\cS(\cH_A)$, where
  $\cH_A\cong\cH_{A_{\mathrm{Z}}}\otimes\cH_{A_{\mathrm{Y}}}$. The
  subsystem $A_{\mathrm{Z}}$ is measured to produce $Z$ and likewise
  for $Y$. The subsystems $A_{\mathrm{Z}}$ and $A_{\mathrm{Y}}$ are
  both considered to be parents of $Z$ (and $Y$). (b) Bell
  scenario. The observed variables $A$ and $B$ together with an
  unobserved system $C$ are used to generate outputs $X$ and $Y$
  respectively. In the classical case, $C$ is modelled as a random
  variable, in the quantum case it is a quantum state on a Hilbert
  space $\cH_C\cong\cH_{C_{\mathrm{X}}}\otimes\cH_{C_{\mathrm{Y}}}$.
  (c) Triangle causal structure. Three observed random variables $X$,
  $Y$ and $Z$ share pairwise common causes $A$, $B$ and $C$, which in
  the classical case are modelled by random variables. Some of the
  valid inequalities such as
  $2 I(X:Y \mid Z)+I(X:Z \mid Y)+I(Y:Z \mid X)-I(X:Y) \geq 0$ can only
  be recovered using non-Shannon entropic
  inequalities~\cite{non_shan}.}
\label{fig:instrumental}
\end{figure}
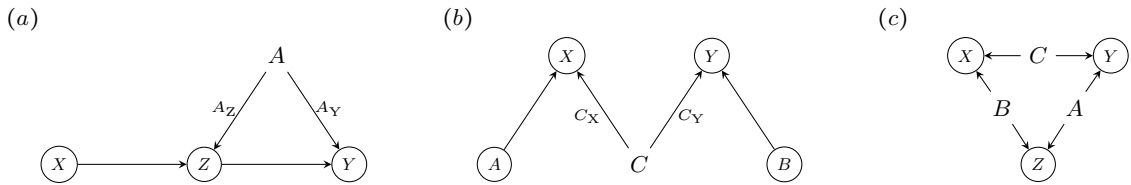

\subsubsection{Classical causal structures}
In the classical case, the causal relations among a set of random
variables can be explored by means of the theory of Bayesian networks
(see for instance~\cite{Spirtes2000,Pearl2009} for a complete
presentation of this theory).
\begin{definition}
  A \emph{classical causal structure}, $C^{\cC}$, is a causal
  structure in which each node of the DAG has an associated random
  variable.
\end{definition}
It is common to use the same label for the node and its associated
random variable. The DAG encodes which joint distributions of the
involved variables are allowed in a causal structure $C^{\cC}$. To
explain this we need a little more terminology.

\begin{definition}
  Let $X_S$, $X_T$, $X_U$ be three disjoint sets of jointly
  distributed random variables. Then $X_S$ and $X_T$ are said to be
  \emph{conditionally independent} given $X_U$ if and only if their
  joint distribution $P_\mathrm{X_S X_T X_U}$ can be written as
  $P_\mathrm{X_S X_T X_U}= P_\mathrm{X_S|X_U} P_\mathrm{X_T|X_U}
  P_\mathrm{X_U}$.  Conditional independence of $X_S$ and $X_T$ given
  $X_U$ is denoted as $X_S \indep X_T | X_U$.
\end{definition}
Two variables $X_S$ and $X_T$ are (unconditionally) independent if
$P_\mathrm{X_S X_T}= P_\mathrm{X_S} P_\mathrm{X_T}$, concisely written
$X_S \indep X_T$.  With reference to a DAG with a subset of nodes,
$X$, we will use $X^{\downarrow}$ to denote the ancestors of $X$ and
$X^{\uparrow}$ to denote the descendants of $X$. The parents of $X$
are represented by $X^{\downarrow_{1}}$ and the non-descendants are
$X^{\nuparrow}$.
\begin{definition}\label{def:compat}
  Let $C^{\cC}$ be a classical causal structure with nodes
  $\left\{X_1,~X_2,~\ldots~,~X_n \right\}$. A probability distribution
  $P_\mathrm{X_1 X_2 \ldots X_n} \in \mathcal{P}_n$ is (Markov)
  \emph{compatible} with $C^{\cC}$ if it can be decomposed as
  $P_\mathrm{X_1 X_2 \ldots X_n}= \prod_{i}
  P_\mathrm{X_{i}|X^{\downarrow_{1}}_{i}}$.
\end{definition}

The compatibility constraint encodes all conditional independences of
the random variables in the causal structure $C^{\cC}$. Nonetheless,
whether a particular set of variables is conditionally independent of
another is more easily read from the DAG, as explained in the
following.

\begin{definition}
  Let $X$, $Y$ and $Z$ be three pairwise disjoint sets of nodes in a
  DAG $G$. The sets $X$ and $Y$ are said to be
  \emph{d-separated}\footnote{The \emph{d} in d-separation stands for
    \emph{directional}.} by $Z$, if $Z$ blocks any path from any node
  in $X$ to any node in $Y$.  A path is \emph{blocked} by $Z$, if the
  path contains one of the following: $i \rightarrow z \rightarrow j$
  or $i \leftarrow z \rightarrow j$ for some nodes $i$, $j$ and a node
  $z \in Z$ in that path, or if the path contains
  $i \rightarrow k \leftarrow j$, where $k \notin Z$.
\end{definition}

The d-separation of the nodes in a causal structure is directly
related to the conditional independence of its variables. The
following proposition corresponds to Theorem~1.2.5
from~\cite{Pearl2009}, previously introduced in~\cite{Verma1988,
  Meek1995}. It justifies the application of d-separation as a means
to identify independent variables.

\begin{prop}[Verma \& Pearl] \label{prop:dseparation} Let $C^{\cC}$ be
  a classical causal structure and let $X$, $Y$ and $Z$ be pairwise
  disjoint subsets of nodes in $C^{\cC}$.  If a probability
  distribution $P$ is compatible with $C^{\cC}$, then the d-separation
  of $X$ and $Y$ by $Z$ implies the conditional independence
  $X \indep Y | Z$.  Conversely, if for every distribution $P$
  compatible with $C^{\cC}$ the conditional independence
  $X \indep Y | Z$ holds, then $X$ is d-separated from $Y$ by $Z$.
\end{prop}

The compatibility of probability distributions with a classical causal
structure is conveniently determined with the following proposition,
which has also been called the parental or local Markov condition
before (Theorem~1.2.7 in~\cite{Pearl2009}).
\begin{prop}[Pearl] \label{prop:localmark}
Let $C^{\cC}$ be a classical causal structure. A probability distribution $P$ is compatible with $C^{\cC}$ if and only if every variable in $C^{\cC}$ is independent of its non-descendants, conditioned on its parents.
\end{prop}
Hence, to establish whether a probability distribution is compatible
with a certain classical causal structure, it is enough to check that
every variable $X$ is independent of its non-descendants
$X^{\nuparrow}$ given its parents $X^{\downarrow_{1}}$, concisely
written as $X \indep X^{\nuparrow} | X^{\downarrow_{1}}$, i.e., to
check one constraint for each variable.  In particular, it is not
necessary to explicitly check for all possible sets of nodes whether
they obey the independence relations implied by d-separation. Each
such constraint can be conveniently expressed as\footnote{This follows
  because the relative entropy
  $D(P \| Q) := \sum_x P_X(x)\log\left(P_X(x)/Q_X(x)\right)$ satisfies
  $D(P\| Q)=0\ \Leftrightarrow \ P = Q$, and because
  $I(X:X^{\nuparrow} | X^{\downarrow_{1}})=D(P_\mathrm{X X^{\nuparrow}
    X^{\downarrow_{1}}} \| P_\mathrm{X | X^{\downarrow_{1}}}
  P_\mathrm{X^{\nuparrow} | X^{\downarrow_{1}}} P_\mathrm{
    X^{\downarrow_{1}}})$.}
\begin{equation} \label{eq:indepentr}
I(X:X^{\nuparrow}| X^{\downarrow_{1}})=0.
\end{equation}

While the conditional independence relations capture some features of
the causal structure, they are insufficient to completely capture the
causal relations between variables, as illustrated in
Figure~\ref{fig:causex}.
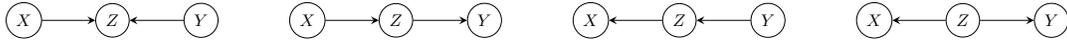
\begin{figure}
\centering 
\resizebox{0.8\columnwidth}{!}{
\begin{tikzpicture}
\node[draw=black,circle,scale=0.75] (1) at (-3.25,0.5) {$X$};
\node[draw=black,circle,scale=0.75] (2) at (-2,0.5) {$Z$};
\node[draw=black,circle,scale=0.75] (3) at (-0.75,0.5) {$Y$};

\node[draw=black,circle,scale=0.75] (4) at (0.75,0.5) {$X$};
\node[draw=black,circle,scale=0.75] (5) at (2,0.5) {$Z$};
\node[draw=black,circle,scale=0.75] (6) at (3.25,0.5) {$Y$};

\node[draw=black,circle,scale=0.75] (7) at (4.75,0.5) {$X$};
\node[draw=black,circle,scale=0.75] (8) at (6,0.5) {$Z$};
\node[draw=black,circle,scale=0.75] (9) at (7.25,0.5) {$Y$};

\node[draw=black,circle,scale=0.75] (10) at (8.75,0.5) {$X$};
\node[draw=black,circle,scale=0.75] (11) at (10,0.5) {$Z$};
\node[draw=black,circle,scale=0.75] (12) at (11.25,0.5) {$Y$};
\draw [->,>=stealth] (1)--(2);
\draw [->,>=stealth] (3)--(2);
\draw [->,>=stealth] (4)--(5);
\draw [->,>=stealth] (5)--(6);
\draw [->,>=stealth] (9)--(8);
\draw [->,>=stealth] (8)--(7);
\draw [->,>=stealth] (11)--(10);
\draw [->,>=stealth] (11)--(12);
\end{tikzpicture}
}
\caption{While in the left causal structure $X \protect\indep Y$, the
  other three networks share the conditional independence relation
  $X \protect\indep Y | Z$. This illustrates that the conditional
  independences are not sufficient to characterise the causal links
  among a set of random variables.}
\label{fig:causex}
\end{figure}
In this case, the probability distributions themselves are unable to
capture the difference between these causal structures: correlations
are insufficient to determine causal links between random variables.
External interventions allow for the exploration of causal links
beyond the conditional independences~\cite{Pearl2009}. However, we do
not consider these here.

Let $C^{\cC}$ be a classical causal structure involving $n$ random
variables $\left\{X_1,~X_2,~\ldots~,~X_n \right\}$.  The restricted
set of distributions that are compatible with the causal structure
$C^{\cC}$ is
${\mathcal{P}\left(C^{\cC}\right):= \left\{P \in \mathcal{P}_n \mid
    P=\prod_{i=1}^{n}P_{\mathrm{X_i | X_{i}^{\downarrow_{1}}}}
  \right\}}.$

\begin{example}[Allowed distributions in the instrumental scenario]
  The classical instrumental scenario of Figure~\ref{fig:instrumental}
  allows for any four variable distribution in the set
  $\mathcal{P}\left(IC^{\cC}\right)={\left\{ P_{\mathrm{AXYZ}} \in
      \mathcal{P}_4 \mid P_{\mathrm{AXYZ}}= P_{\mathrm{Y | AZ}}
      P_{\mathrm{Z | AX}} P_{\mathrm{X}} P_{\mathrm{A}} \right\}.}$
\end{example}

The restrictions on the allowed distributions also restrict the
corresponding entropy cones. Due to Proposition~\ref{prop:localmark}
there are at most $n$ independent conditional independence equalities
\eqref{eq:indepentr} in a causal structure $C^{\cC}$. Their
coefficients can be concisely written in terms of a matrix
$M_\mathrm{CI}\left(C^{\cC}\right)$, where CI stands for conditional
independence.  For a causal structure $C^{\cC}$, we define the two
sets
$\Gamma^{*}\left(C^{\cC}\right):= \left\{ v \in \Gamma^{*}_n \mid
  M_\mathrm{CI}\left(C^{\cC}\right) \cdot v = 0 \right\}$ and
$\Gamma\left(C^{\cC}\right):= \left\{ v \in \Gamma_n \mid
  M_\mathrm{CI}\left(C^{\cC}\right) \cdot v = 0 \right\}$, where
$\Gamma^{*}\left(C^{\cC}\right) \subseteq \Gamma\left(C^{\cC}\right)$.
The following lemma justifies the notation we use for
$\Gamma^{*}\left(C^{\cC}\right)$; it is the set of achievable entropy
vectors in $C^{\cC}$.

\begin{lemma} \label{lemma:convexity} For a causal structure
  $C^{\cC}$,
  $ \Gamma^{*} \left(C^{\cC}\right)\!= \! \left\{v \in
    \mathbb{R}_{\geq 0}^{2^{n}-1} \mid \exists P \!= \!
    \prod_{i=1}^{n} P_{\mathrm{X_i |
        X_{i}^{\downarrow_{1}}}}\in\cP_n\text{ s.t.\
    }v\!=\!{\bf{H}}(P) \!\right\}.$ Furthermore, its topological
  closure,
  ${ \overline{\Gamma^*}\left(C^{\cC}\right)=\left\{ v \in
      \overline{\Gamma^*_n}\mid M_\mathrm{CI}\left(C^{\cC}\right)
      \cdot v = 0 \right\}, }$ is a convex cone.
\end{lemma}
\begin{proof}
  For the causal structure $C^{\cC}$, let
  ${ E \left(C^{\cC}\right)\!:= \! \left\{v \in \mathbb{R}_{\geq
        0}^{2^{n}-1} \mid \exists P \!= \! \prod_{i=1}^{n}
      P_{\mathrm{X_i | X_{i}^{\downarrow_{1}}}}\in\cP_n\text{ s.t.\
      }v\!=\!{\bf{H}}(P) \!\right\}}$ be the set of all entropy
  vectors.  Since~\eqref{eq:indepentr} holds for each variable $X_i$
  if and only if
  $P= \prod_{i=1}^{n} P_{\mathrm{X_i | X_{i}^{\downarrow_{1}}}}$ (cf.\
  Proposition~\ref{prop:localmark}),
  ${ E \left(C^{\cC}\right)\!= \!\left\{v \in \Gamma^{*}_n \mid
      \forall i \in \left\{1,~\ldots~,n\right\},~ I(X_i :
      X_{i}^{\nuparrow} |X_{i}^{\downarrow_{1}})=0 \! \right\}}$.
  Applying the definition of $M_\mathrm{CI}\left(C^{\cC}\right)$
  yields $E \left(C^{\cC}\right)=\Gamma^{*}\left(C^{\cC}\right)$.
  Now, let us consider the set
  ${ F(C^{\cC}):=\left\{ v \in \overline{\Gamma^*_n}\mid
      M_\mathrm{CI}\left(C^{\cC}\right) \cdot v = 0 \right\} \subseteq
    \mathbb{R}^{2^{n}-1} }$.  This is a closed convex set, since
  $\overline{\Gamma^*_n}$ is known to be closed and convex and since
  restricting the closed convex cone $\overline{\Gamma^*_n}$ with
  linear equality constraints retains these properties. More
  precisely, the set of solutions to the matrix equality
  $M_\mathrm{CI}\left(C^{\cC}\right) \cdot v = 0$ is also closed and
  convex. Being the intersection of two closed convex sets, the set
  $F(C^{\cC})$ is also closed and convex.  From this we conclude that
  $\overline{\Gamma^*}\left(C^{\cC}\right)$ is convex because
  ${\overline{\Gamma^*}\left(C^{\cC}\right)=\overline{\left\{ v \in
        \Gamma^{*}_n \mid M_\mathrm{CI}\left(C^{\cC}\right) \cdot v =
        0 \right\}} }$ equals $F(C^{\cC})$. (Because $F(C^{\cC})$ is
  closed, any element $w \in F(C^{\cC})$, in particular any element on
  its boundary, is the limit of a sequence of elements
  $\left\{w_k \right\}_k$ for $k \rightarrow \infty$, where the $w_k$
  lie in the interior of $F(C^{\cC})$ for all $k$. Hence
  $w \in \overline{\Gamma^*}\left(C^{\cC}\right)$.)
\end{proof}
The convexity of $\overline{\Gamma^*}\left(C^{\cC}\right)$ is crucial
for the considerations of the following sections.  Note that in spite
of the convexity of $\overline{\Gamma^*}\left(C^{\cC}\right)$, the set
$\mathcal{P}\left(C^{\cC}\right)$ is generally not convex. This
alludes to the fact that significant information about the achievable
correlations among the random variables is lost via the mapping from
$\mathcal{P}\left(C^{\cC}\right)$ to the corresponding entropic cone
$\overline{\Gamma^*}\left(C^{\cC}\right)$.

\begin{example}[Entropic outer approximation for the instrumental
  scenario]
The instrumental scenario has at most $4$ independent conditional independence equalities~\eqref{eq:indepentr}. We find that there are only two,
${I(A:X)=0}$ and ${I(Y:X\mid AZ)=0}$.
This yields
$\overline{\Gamma^{*}}\left(IC^{\cC}\right)= \left\{ v \in \overline{\Gamma^{*}_4} \mid
  M_\mathrm{CI}\left(IC^{\cC}\right) \cdot v = 0 \right\}$ with
$$ M_\mathrm{CI}\left(IC^{\cC}\right)=\left( \begin{array}{ccccccccccccccc}
-1 & -1 & 0 & 0 & 1 & 0 & 0 & 0 & 0 & 0 & 0 & 0 & 0 & 0 & 0 \\
0 & 0 & 0 & 0 & 0 & 0 & 1 & 0 & 0 & 0 & 0 & -1 & -1 & 0 & 1
\end{array} \right), $$
where the coordinates are ordered as
$(H(A),H(X),H(Y),H(Z),H(AX),H(AY),H(AZ),H(XY),H(XZ),H(YZ),\\ H(AXY),H(AXZ),H(AYZ),H(XYZ),H(AXYZ)).$  
An outer approximation is given by $\Gamma\left(IC^{\cC}\right)=\left\{ v \in \Gamma_4 \mid
  M_\mathrm{CI}\left(IC^{\cC}\right) \cdot v = 0 \right\}.$
\end{example}

In general, the outer approximation to
$\overline{\Gamma^{*}}\left(C^{\cC}\right)$ can be further tightened
by taking non-Shannon inequalities into account. These have lead to
the derivation of numerous new entropic inequalities for various
causal structures~\cite{non_shan} (see e.g.\ the triangle causal
structure of Figure~\ref{fig:instrumental}(c)). For the instrumental
scenario, however, such additional inequalities are irrelevant. This
can for instance be seen by constructing the following inner
approximation to the cone.
\begin{example}[Entropic inner approximation for the instrumental scenario~\cite{non_shan}]
  For the instrumental scenario an inner approximation is given in
  terms of the Ingleton cone and the conditional independence
  constraints from the previous example,
  ${\Gamma^{\iI}\left(IC^{\cC}\right)}= {\left\{v \in \Gamma^{\iI}
      \mid M_\mathrm{CI}\left(IC^{\cC}\right) \cdot v = 0 \right\} }$.
  For this causal structure the Ingleton constraints are implied by
  the Shannon inequalities and the conditional independence
  constraints and hence, inner and outer approximation
  coincide. Consequently, they also coincide with the actual entropy
  cone, i.e.,
  $\Gamma^{\iI}\left(IC^{\cC}\right)= \Gamma\left(IC^{\cC}\right)=
  \overline{\Gamma^{*}}\left(IC^{\cC}\right)$. In particular,
  non-Shannon entropic inequalities cannot improve the outer
  approximation in this example.
\end{example}
Inner approximations have been considered in~\cite{non_shan}. They are
particularly useful in cases where identical inner and outer
approximations are found, where they identify the actual boundary of
the entropy cone. In other cases they can allow parts of the actual
boundary to be identified or give clues on how to find better outer
approximations.

Arguably all interesting scenarios (such as the previous example)
involve unobserved variables that are suspected to cause some of the
correlations between the variables we observe.  These unobserved
variables may yield constraints on the possible joint distributions of
the observed variables, a well-known example being a Bell
inequality~\cite{Bell1964}\footnote{For a detailed discussion of the
significance of Bell inequality violation on classical causal
structures see~\cite{Wood2012}.}.  More generally we would like to
infer constraints on the observed variables that follow from the
presence of unobserved variables.

For a causal structure on $n$ random variables
$\left\{X_1,~X_2,~\ldots~,~X_n \right\}$, the restriction to the set
of observed variables is called its \emph{marginal scenario}, denoted
$\mathcal{M}$.  Here, we assume w.l.o.g.\ that the first $k\leq n$
variables are observed and the remaining $n-k$ are not. We are thus
interested in the correlations among the first $k$ variables that can
be obtained as the marginal of some distribution over all $n$
variables.  Without any causal restrictions the set of all probability
distributions of the $k$ observed variables is
${\mathcal{P}_{\mathcal{M}} := \left\{P \in \mathcal{P}_k \mid
    P=\sum_{X_{k+1},\ldots,X_n} P_\mathrm{X_1 X_2 \ldots X_n} \right\}
},$ i.e., $\mathcal{P}_{\mathcal{M}} = \mathcal{P}_k$.  For a
classical causal structure $C^{\cC}$ on the set of variables
$\left\{X_1,~ X_2,~\ldots~,~X_n \right\}$, marginalising all
distributions $P \in \mathcal{P}\left(C^{\cC}\right)$ over the $n-k$
unobserved variables leads to the set
$\mathcal{P}_\mathcal{M}\left(C^{\cC}\right) := \left\{ P \in
  \mathcal{P}_k \mid P=\sum_{X_{k+1},\ldots,X_n} \prod_{i=1}^{n}
  P_\mathrm{X_i | X_{i}^{\downarrow_{1}}} \right\}$.  In contrast to
the unrestricted case, this set of distributions is in general not
recovered by considering a causal structure that involves only $k$
observed random variables, as can be seen in the following example.
\begin{example}[Observed distributions in the instrumental scenario]
For the instrumental scenario the observed variables are $X$, $Y$ and $Z$ and their joint distribution is of the form
${\mathcal{P}_\mathcal{M}\left(IC^{\cC}\right)= \left\{ P_{\mathrm{XYZ}} \in \mathcal{P}_3 \mid P_{\mathrm{XYZ}}= \sum_A P_{\mathrm{Y | AZ}} P_{\mathrm{Z | AX}} P_{\mathrm{X}} P_{\mathrm{A}} \right\} }$.
\end{example}

The first entropic inequalities for a marginal scenario were
derived in~\cite{Steudel2015}, where certificates for the existence of
common ancestors of a subset of the observed random variables of at
least a certain size were given. One such scenario is the triangle
causal structure of Figure~\ref{fig:instrumental}(c).  The systematic
entropy vector approach was devised for classical causal structures
in~\cite{Chaves2012, Fritz2013, Chaves2014b}. An outer approximation
to the entropic cones of a variety of causal structures was given
in~\cite{Henson2014}. In the following we give the details of this
approach.

In the entropic picture, marginalisation is performed by eliminating
the coordinates that represent entropies of sets of variables
containing at least one unobserved variable from the vectors.  This
corresponds to a projection of a cone in $\mathbb{R}^{2^{n}-1}$ to its
marginal cone in $\mathbb{R}^{2^{k}-1}$~\cite{Chaves2014}. We will
denote this projection
$\pi_\mathcal{M}:\mathbb{R}^{2^{n}-1}\rightarrow\mathbb{R}^{2^{k}-1}$.
It gives all entropy vectors $w$ of the observed sets of variables,
i.e., of the marginal scenario $\mathcal{M}$, for which there exists
at least one entropy vector $v$ in the original scenario with matching
entropies on the observed variables.

Starting from the set of all entropy vectors,
$\Gamma^{*}\left(C^{\cC}\right)$, those relevant for the marginal
scenario can be obtained by discarding the appropriate components.
For a finitely generated cone such as $\Gamma\left(C^{\cC}\right)$,
its projection can be more efficiently determined from the projection
of its extremal rays.  In the dual description of the entropic cone in
terms of its facets (i.e., its inequality description) the transition
to the marginal scenario can be made computationally by eliminating
all entropies of sets of variables not contained in $\mathcal{M}$ from
the system of inequalities. The standard algorithm that achieves this
is Fourier-Motzkin elimination~\cite{Williams1986}, which has been
used in this context in~\cite{Chaves2012, Fritz2013, Chaves2014}.

Without any causal restrictions, the entropy cone
$\overline{\Gamma^{*}_n}$ is projected to the marginal cone
$\overline{\Gamma^{*}_{\mathcal{M}}}:=\left\{w \in
  \mathbb{R}^{2^k-1}_{\geq 0} \mid \exists v \in
  \overline{\Gamma^{*}_n} \text{ s.t.\ } w= \pi_\mathcal{M}(v)
\right\}$. Note that if we marginalise $\overline{\Gamma^{*}_n}$ over
$n-k$ variables we recover the entropy cone for $k$ random variables,
i.e.,
$ \overline{\Gamma^{*}_{\mathcal{M}}} = \overline{\Gamma^{*}_{k}} $.
The same applies to the outer approximations: The $n$ variable Shannon
cone $\Gamma_{n}$ is projected to the $k$ variable Shannon cone with
the mapping $\pi_\mathcal{M}$, i.e.,
$\Gamma_{\mathcal{M}} := \left\{ w \in \mathbb{R}^{2^k-1}_{\geq 0}
  \mid \exists v \in \Gamma_n \text{ s.t.\ } w= \pi_\mathcal{M}(v)
\right\} = \Gamma_k$. This follows because the $n$ variable Shannon
constraints contain the corresponding $k$ variable constraints as a
subset, and since any vector in $\Gamma_k$ can be extended to a vector
in $\Gamma_n$, for instance by taking
$H(X_{k+1})=H(X_{k+2})= \cdots = H(X_n)=0$ and
$H(X_S \cup X_T)=H(X_S)$ for any
$X_T \subseteq \left\{ X_{k+1},~X_{k+2},~\ldots~,~X_n \right\}$.

For a classical causal structure $C^{\cC}$, we will be interested in
the set
$\overline{\Gamma^{*} _{\mathcal{M}}}\left(C^{\cC}\right):=\left\{w
  \in \mathbb{R}^{2^k-1}_{\geq 0} \mid \exists v \in
  \overline{\Gamma^{*}}\left(C^{\cC}\right) \mathrm{ s.t. }  w=
  \pi_\mathcal{M}(v) \right\}$, which is by construction a convex
cone, since projection preserves convexity.  The following lemma
confirms that this is the entropy cone of the marginal scenario
$\mathcal{M}$ and thus also formally justifies the method of
projecting the sets directly.

\begin{lemma}
  $\overline{\Gamma^{*} _{\mathcal{M}}}\left(C^{\cC}\right)$ is equal
  to the set of entropy vectors compatible with the marginal scenario
  of the classical causal structure $C^{\cC}$, i.e.,
  ${
    \overline{\Gamma^{*}_{\mathcal{M}}}\left(C^{\cC}\right)=\overline{\left\{
        w \in \mathbb{R}^{2^{k}-1}_{\geq 0} \mid \exists P \in
        \mathcal{P}_\mathcal{M}\left(C^{\cC}\right) \mathrm{ s.t. }
        w={\bf{H}}(P) \right\} }} $.
\end{lemma}
\begin{proof}
  Let $\cF$ denote the set on the rhs of the statement in the lemma.
  Note that $w\in\Gamma^{*}_{\mathcal{M}}\left(C^{\cC}\right)$ implies
  that there exists $v\in\Gamma^{*}\left(C^{\cC}\right)$ s.t.\
  $w=\pi_\mathcal{M}(v)$.  Using Lemma~\ref{lemma:convexity}, we have
  $v={\bf H}(P)$ for some
  $P=\prod_iP_\mathrm{X_i|X_i^{\downarrow_{1}}}$.  If we take
  $P'=\sum_{X_{k+1},\ldots,X_n}P$ then $w={\bf H}(P')$ and hence
  $w\in\cF$.

  Conversely, $w\in\cF$ implies that there exists
  $P'\in\mathcal{P}_\mathcal{M}\left(C^{\cC}\right)$ s.t.\
  $w={\bf H}(P')$ and hence there exists
  $P\in\mathcal{P}\left(C^\cC \right)$ such that
  $P'=\sum_{X_{k+1},\ldots,X_n}P$.  If we take $v={\bf H}(P)$, then
  $w=\pi_\mathcal{M}(v)$ and hence
  $w\in\Gamma^{*} _{\mathcal{M}}\left(C^{\cC}\right)$. Taking the
  topological closure of both sets concludes the proof.
\end{proof}
An outer approximation to
$\overline{\Gamma^{*}_{\mathcal{M}}}\left(C^{\cC}\right)$ is
$\Gamma_{\mathcal{M}}\left(C^{\cC}\right):=\left\{w \in
  \mathbb{R}^{2^k-1}_{\geq 0} \mid \exists v \in
  \Gamma\left(C^{\cC}\right) \mathrm{ s.t. }  w= \pi_\mathcal{M}(v)
\right\}$, which can be written as
$\Gamma_{\mathcal{M}}\left(C^{\cC}\right)=\left\{w \in
  \Gamma_\mathcal{M} \mid M_\mathcal{M}\left(C^{\cC}\right) \cdot w
  \geq 0 \right\}$, where $M_\mathcal{M}\left(C^{\cC}\right)$ is the
matrix obtained via Fourier-Motzkin elimination and encodes the set of
inequalities on $\mathbb{R}^{2^k-1}_{\geq 0}$ that are implied by the
fact that $M_\mathrm{CI}\left(C^{\cC}\right)\cdot v=0$ and
$M_\mathrm{SH}^n\cdot v\geq 0$ hold on $\mathbb{R}^{2^n-1}_{\geq 0}$
(except for the $k$-variable Shannon constraints which are already
included in $\Gamma_\mathcal{M}$).

\begin{example}[Entropic outer approximation for the marginal cone of
  the instrumental
  scenario~\cite{Henson2014,Chaves2014b}] \label{example:IC_class_obs}
  For the instrumental scenario the outer approximation to its
  marginal cone is found by projecting $\Gamma\left(IC^{\cC}\right)$
  to its three variable scenario and yields
  $\Gamma_{\mathcal{M}}\left(IC^{\cC}\right) =\left\{w \in
    \Gamma_\mathcal{M} \mid M_\mathcal{M}\left(IC^{\cC}\right) \cdot w
    \geq 0 \right\}$, where
  $M_\mathcal{M}\left(IC^{\cC}\right)=\left( \begin{array}{ccccccc} -1
      & 0 & 1 & 0 & 0 & -1 & 1 \end{array} \right)$ corresponds to the
  inequality $I(X:YZ) \leq H(Z)$ from~\cite{Henson2014,Chaves2014b}.
  As
  $\Gamma\left(IC^{\cC}\right)=\overline{\Gamma^{*}}\left(IC^{\cC}\right)$
  holds, we also have
  $\Gamma_{\mathcal{M}}\left(IC^{\cC}\right)=\overline{\Gamma^{*}_{\mathcal{M}}}\left(IC^{\cC}\right)$.
\end{example}
As mentioned previously, non-Shannon inequalities cannot give any new
entropic constraints for $IC$, as the Shannon approximation is already
tight. However, in many causal structures they do. For instance in the
triangle scenario of Figure~\ref{fig:instrumental}(c), non-Shannon
inequalities still lead to new entropic constraints, even after
marginalisation to the three observed variables~\cite{non_shan}.

\subsubsection{Causal structures with unobserved quantum
  systems} \label{sec:quantumlatent}
A quantum causal structure differs from its classical counterpart in
that unobserved systems correspond to shared quantum states.

\begin{definition}
  A \emph{quantum causal structure}, $C^{\qQ}$, is a causal structure
  where each observed node has a corresponding random variable, and
  each unobserved node has an associated quantum system.
\end{definition}

In a classical causal structure the edges of the DAG represent the
propagation of classical information, and, at a node with incoming
edges, the random variable there can be generated by applying an
arbitrary function to its parents.  We are hence implicitly assuming
that all the information about the parents is transmitted to its
children (otherwise the set of allowed functions would be restricted).
This does not pose a problem since classical information can be
copied.  In the quantum case, on the other hand, the no-cloning
theorem means that the children of a node cannot (in general) all have
access to the same information as is present at that node.
Furthermore, the analogue of performing arbitrary functions in the
classical case is replaced by arbitrary quantum operations.  Such a
quantum framework that allows for an analysis with entropy vectors was
introduced in~\cite{Chaves2015}. In the following we outline this
approach. However, for unity of description, our account of quantum
causal structures is based upon the viewpoint that is taken for
generalised causal structures in~\cite{Henson2014}, which we review in
the next section\footnote{The difference is as follows: In~\cite{Chaves2015}
nodes correspond to quantum systems. All outgoing edges of a node
together define a completely positive trace preserving (CPTP) map with
output states corresponding to the joint state associated with its
child nodes. Similarly, the CPTP map associated to the input edges of
a node must map the states of the parent nodes to the node in
question. In~\cite{Henson2014}, on the other hand, edges correspond to
states whereas the transformations occur at the nodes.}.

Let $C^{\qQ}$ be a quantum causal structure. Nodes without input edges
correspond to the preparation of a quantum state described by a
density operator on a Hilbert space, e.g., $\rho_A \in \cS(\cH_A)$ for
a node $A$, where for observed nodes this state is required to be
classical\footnote{$\cS(\cH)$ denotes the set of all density operators
  on a Hilbert space $\cH$.}.  For each directed edge in the graph
there is a corresponding subsystem with Hilbert space labelled by the
edge's input and output nodes. For instance, if $Y$ and $Z$ are the
only children of $A$ then there are associated spaces $\cH_{A_Y}$ and
$\cH_{A_Z}$ such that $\cH_A=\cH_{A_Y}\otimes\cH_{A_Z}$\footnote{Note
  that in the classical case these subsystems may all be taken to be
  copies of the system itself.}. At an unobserved node, a CPTP map from
the joint state of all its input edges to the joint state of its
output edges is performed. A node is labelled by its output state. For
an observed node the latter is classical. Hence, it corresponds to a
random variable that represents the output statistics obtained in a
measurement by applying a positive operator valued measure (POVM) to
the input states\footnote{Note that preparation and measurement can
  also be seen as CPTP maps with classical input and output systems
  respectively, thus allowing for a unified formulation.}. If all
input edges are classical this can be interpreted as a stochastic map
between random variables.

A distribution, $P$, over the observed nodes of a causal structure
$C^{\qQ}$ is compatible with $C^{\qQ}$ if there exists a quantum state
labelling each unobserved node (with subsystems for each unobserved
edge) and transformations, i.e., preparations and CPTP maps for each
unobserved node as well as POVMs for each observed node, that allow
for the generation of $P$ by means of the Born rule. We denote the set
of all compatible distributions
$\mathcal{P}_{\mathcal{M}}\left( C^{\qQ} \right)$.

\begin{example}[Compatible distributions in the quantum instrumental scenario]
  For the quantum instrumental scenario
  (Figure~\ref{fig:instrumental}(a)),
  ${\mathcal{P}_{\mathcal{M}}\left( IC^{\qQ} \right)=\left\{ P_{XYZ}
      \in \cP_3 \mid P_{XYZ}=\operatorname{tr}((E_X^{Z} \otimes
      F_Z^{Y}) \rho_{A}) P_X \right\} }$ is the set of compatible
  distributions.  A state
  $\rho_A \in \mathcal{S}(\cH_{A_Z}\otimes \cH_{A_Y})$ is
  prepared. Depending on the random variable $X$, a POVM
  $\left\{E_X^{Z} \right\}_Z$ on $\cH_{A_Z}$ is applied to generate
  the output distribution of the observed variable $Z$. Depending on
  the latter, another POVM $\left\{ F_Z^{Y}\right\}_Y$ is applied to
  generate the distribution of $Y$.
\end{example}

The set of entropy vectors of compatible probability distributions
over the observed nodes,
$\mathcal{P}_{\mathcal{M}}\left(C^{\qQ}\right)$, is
$\Gamma^{*}_{\mathcal{M}}\left( C^{\qQ} \right) := \left\{w \in
  \mathbb{R}^{2^{k}-1}_{\geq 0} \mid \exists P \in
  \mathcal{P}_{\mathcal{M}}\left( C^{\qQ} \right) \mathrm{ s.t. }  w=
  {\bf{H}}(P) \right\}$.  Outer approximations
$\Gamma_\mathcal{M} \left(C^{\qQ} \right)$ were first derived
in~\cite{Chaves2015}, a procedure that we outline in the following.
For their construction, an entropy is associated to each random
variable and to each subsystem of a quantum state (equivalently each
edge originating at a quantum node), corresponding to the von Neumann
entropy of the respective system. For convenience of exposition, edges
and their associated systems share the same label.  The \emph{von
  Neumann entropy} of a density operator $\rho_{X} \in \cS(\cH_X)$ is
defined as
$$H(X):=-\operatorname{tr}(\rho_X \log_2 ( \rho_X))\, .$$
The quantum conditional entropy, mutual- and conditional mutual
information are defined as in the classical case and with the von
Neumann entropy replacing the Shannon entropy.

Because of the impossibility of cloning, the outcomes and the quantum
systems that led to them do not exist simultaneously. Therefore there
is in general no joint multi-party quantum state for all subsystems
and it does not make sense to talk about the joint entropy of the
states and outcomes.  More concretely, if a system $A$ is measured to
produce $Z$, then $\rho_{AZ}$ is not defined and hence neither is
$H(AZ)$\footnote{Attempts to circumvent this have been made, see for
  example~\cite{Leifer2013}.}.
\begin{definition}
  Two subsystems in a quantum causal structure $C^{\qQ}$
  \emph{coexist} if neither of them is a quantum ancestor of the
  other. A set of subsystems that mutually coexist is termed
  \emph{coexisting}.
\end{definition}\label{def:coex}
A quantum causal structure may have several maximal coexisting subsets.  Only within such subsets is there a well defined joint quantum state and joint entropy.

\begin{example}[Coexisting sets in the quantum instrumental scenario]\label{example:coexsets}
  Consider the quantum version of the instrumental scenario, as
  illustrated in Figure~\ref{fig:instrumental}(a). There are three
  observed variables as well as two edges originating at unobserved
  (quantum) nodes, hence $5$ variables to consider. More precisely,
  the quantum node $A$ has two associated subsystems $A_{\mathrm{Z}}$
  and $A_{\mathrm{Y}}$. The correlations seen at the two observed
  nodes $Z$ and $Y$ are formed by measurement on the respective
  subsystems $A_{\mathrm{Z}}$ and $A_{\mathrm{Y}}$. The coexisting
  sets in this causal structure are
  $\left\{ A_{\mathrm{Y}},~A_{\mathrm{Z}},~X \right\}$,
  $\left\{ A_{\mathrm{Y}},~X,~Z \right\}$ and
  $\left\{ X,~Y,~Z \right\}$ and their (non-empty) proper subsets.
\end{example}

Note that without loss of generality we can assume that any initial,
i.e., parentless quantum states such as $\rho_A$ above, are pure. This
is because any mixed state can be purified, and if the transformations
and measurement operators are then taken to act trivially on the
purifying systems the same statistics are observed.  In the causal
structure of Example~\ref{example:coexsets}, this implies that
$\rho_A$ can be considered to be pure and thus
$H(A_{\mathrm{Y}} A_{\mathrm{Z}})=0$. The Schmidt decomposition then
implies that $H(A_{\mathrm{Y}})=H(A_{\mathrm{Z}})$. This is
computationally useful as it reduces the number of free parameters in
the entropic description of the scenario.  Furthermore, by
Stinespring's theorem~\cite{Stinespring}, whenever a CPTP map is
applied at a node that has at least one quantum child, then one can
instead consider an isometry to a larger output system. The additional
system that is required for this can be taken to be part of the
unobserved quantum output (or one of them in case of several quantum
output nodes). Each such case allows for the reduction of the number
of variables by one, since the joint entropy of all inputs to such a
node must be equal to that of all its outputs.

Quantum states are known to obey submodularity~\cite{Lieb1973} and
also obey the condition
\begin{itemize}
\item Weak monotonicity~\cite{Lieb1973}:
  $H(X_S\setminus X_T) + H(X_T\setminus X_S) \leq H(X_S)+H(X_T)$, for
  all $X_S$, $X_T\subseteq \Omega$ (recall $H(\left\{ \right\})=0$).
\end{itemize} 
This the dual of submodularity in the sense that the two inequalities
can be derived from each other by considering purifications of the
corresponding quantum states~\cite{Araki1970}.

Within the context of causal structures, these relations can always be
applied between variables in the same coexisting set. In addition,
whenever it is impossible for there to be entanglement between the
subsystems $X_S\cap X_T$ and $X_S\setminus X_T$ --- for instance if
these subsystems are in a cq-state --- the monotonicity constraint
$H(X_S\setminus X_T)\leq H(X_S)$ holds.  If it is also impossible for
there to be entanglement between $X_S\cap X_T$ and $X_T\setminus X_S$,
then the monotonicity relation $H(X_T\setminus X_S)\leq H(X_T)$ holds
rendering the weak monotonicity relation stated above redundant.

Altogether, these considerations lead to a set of basic inequalities
containing some Shannon and some weak-monotonicity inequalities, which
are conveniently expressed in a matrix
$M_\mathrm{B}\left(C^{\qQ}\right)$. This way of approximating the
entropic cone in the quantum case is inspired by work on the entropic
cone for multi-party quantum states~\cite{Pippenger2003}.  Note also
that there are no further inequalities for the von Neumann entropy
known to date (contrary to the classical case where a variety of non
Shannon inequalities is known), except under additional
constraints~\cite{Linden2005,Cadney2012,Linden2013,Gross2013,Bao1,Bao2}.

The conditional independence constraints in $C^{\qQ}$ cannot be
identified by Proposition~\ref{prop:localmark}, because variables do
not coexist with any quantum parents and hence conditioning a variable
on a quantum parent is not meaningful.  Nonetheless, among the
variables in a coexisting set the conditional independences that are
valid for $C^{\cC}$ also hold in $C^{\qQ}$.  This can be seen as
follows. First, the validity of any constraints that involve only
observed variables (which are always part of a coexisting set) hold by
Proposition~\ref{prop:dseparationgDAG} below.  Secondly, for
unobserved systems only their classical ancestors and none of their
descendants can be part of the same coexisting set. An unobserved
system is hence independent of any subset of the same coexisting set
with which it shares no ancestors. Note that each of the subsystems
associated with a quantum node is considered to be a parent of all of
the node's children (see Figure~\ref{fig:instrumental} for an
example).

In addition, suppose $X_S$ and $X_T$ are disjoint subsets of a
coexisting set, $\Xi$, and that the unobserved system $A$ is also in
$\Xi$.  Then $I(A:X_S|X_T)=0$ if $X_T$ d-separates $A$ from $X_S$
(in the full graph including quantum nodes)\footnote{This follows because any
quantum states generated from the classical separating variables may
be obtained by first producing random variables from the latter (for
which the usual d-separation rules hold) and then using these to
generate the quantum states in question (potentially after generating
other variables in the network), hence retaining conditional
independence.}.  The same considerations can be made for sets of
unobserved systems.  These independence constraints may be assembled
in a matrix $M_\mathrm{QCI}\left(C^{\qQ}\right)$.

Among the variables that do not coexist, some are obtained from others
by means of quantum operations. These variables are thus related by
data processing inequalities~(DPIs)~\cite{BookNielsenChuang2000}.
\begin{prop}[DPI] \label{prop:DPI} Let
  $\rho_\mathrm{X_S X_T} \in \mathcal{S}(\mathcal{H}_{X_\mathrm{S}}
  \otimes \mathcal{H}_{X_{\mathrm{T}}})$
  and $\mathcal{E}$ be a completely positive trace preserving (CPTP)
  map\footnote{Note that the map from a quantum state to the diagonal
    state with entries equal to the outcome probabilities of a
    measurement is a CPTP map and hence also obeys the DPI.} on
  $\mathcal{S}(\mathcal{H}_{X_{\mathrm{T}}})$\footnote{In general
    $\mathcal{E}$ can be a map between operators on different Hilbert
    spaces, i.e.,\
    $\mathcal{E}:\mathcal{S}(\mathcal{H}'_{X_{\mathrm{T}}})
    \rightarrow \mathcal{S}(\mathcal{H}''_{X_{\mathrm{T}}})$.
    However, as we can consider these operators to act on the same
    larger Hilbert space we can w.l.o.g.\ take $\mathcal{E}$ to be a
    map on this larger space, which we call
    $\mathcal{S}(\mathcal{H}_{X_{\mathrm{T}}})$.} leading to a state
  $\rho'_\mathrm{X_{\mathrm{S}} X_{\mathrm{T}}}$.  Then
  ${I(X_{\mathrm{S}} : X_{\mathrm{T}})_{\rho'_\mathrm{X_S X_T}} \leq
    I(X_{\mathrm{S}} : X_{\mathrm{T}})_{\rho_\mathrm{X_S X_T}}}$.
\end{prop}

The data processing inequalities provide an additional set of entropic
constraints, which can be expressed in terms of a matrix inequality
$M_\mathrm{DPI}\left(C^{\qQ}\right) \cdot v \geq 0$\footnote{There are
  also DPIs for conditional mutual information, e.g.,
  $I(A:B|C)_{\rho'_{ABC}}\leq I(A:B|C)_{\rho_{ABC}}$ for
  $\rho'_{ABC}=(\cI\otimes\cE\otimes\cI)(\rho_{ABC})$, but these are
  implied by Proposition~\ref{prop:DPI}, so they need not be treated
  separately here.}.  In general, there are a large number of
variables for which data processing inequalities hold. It is thus
beneficial to derive rules that specify which of the inequalities are
needed.  First, note that whenever a concatenation of two CPTP maps
$\mathcal{E}_1$ and $\mathcal{E}_2$,
$\mathcal{E}=\mathcal{E}_2 \circ \mathcal{E}_1$, is applied to a
state, then any DPIs for inputs and outputs of $\mathcal{E}$ are
implied by the DPIs for $\mathcal{E}_1$ and $\mathcal{E}_2$\footnote{This
follows by deriving the DPIs for input and output states of
$\mathcal{E}_1$ and $\mathcal{E}_2$ respectively and combining the
two.}. Hence, the DPIs for composed maps $\mathcal{E}$ never have to be
considered as separate constraints.

Secondly, whenever a state
$\rho_\mathrm{X_S X_T X_R} \in \mathcal{S}(\mathcal{H}_{X_\mathrm{S}}
\otimes \mathcal{H}_{X_{\mathrm{T}}}\otimes
\mathcal{H}_{X_{\mathrm{R}}})$
can be decomposed as
$\rho_\mathrm{X_S X_T X_R}=\rho_\mathrm{X_S X_T}\otimes
\rho_\mathrm{X_R}$
and a CPTP map $\mathcal{E}$ transforms the state on
$\mathcal{S}(\mathcal{H}_{X_{\mathrm{S}}})$. Then any DPIs for
$\rho_\mathrm{X_S X_T X_R}$ are implied by the DPIs for
$\rho_\mathrm{X_S X_T}$\footnote{This follows from
  $I(X_{\mathrm{S}}:X_{\mathrm{T}}X_{\mathrm{R}})=I(X_{\mathrm{S}}:X_{\mathrm{T}})$,
  $I(X_{\mathrm{S}}X_{\mathrm{R}}:X_{\mathrm{T}})=I(X_{\mathrm{S}}:X_{\mathrm{T}})$
  and $I(X_{\mathrm{S}}X_{\mathrm{T}}:X_{\mathrm{R}})=0$.}.

Furthermore, whenever a node has classical and quantum inputs, there
is not only a CPTP map generating its output state but this map can be
extended to a CPTP map that simultaneously retains the classical
inputs, as is the content of the following lemma, which also shows
that retaining a copy of the classical inputs leads to tighter
entropic inequalities.
\begin{lemma}\label{lemma:simplify_DPI_classicalquantum}
Let $Y$ be a node with classical and quantum inputs
$X_\mathrm{C}$ and $X_\mathrm{Q}$ and $\cE$ be a CPTP map that acts at
this node, i.e., $\cE$ is a map from
$\cS(\cH_{X_\mathrm{C}}\otimes\cH_{X_\mathrm{Q}})$ to $\cS(\cH_Y)$. Then
$\cE$ can be extended to a map
$\cE':\cS(\cH_{X_\mathrm{C}}\otimes\cH_{X_\mathrm{Q}})\to\cS(\cH_{X_\mathrm{C}}\otimes\cH_Y)$
such that
$\cE':\rho_{X_\mathrm{C}X_\mathrm{Q}}\mapsto\rho'_{X_\mathrm{C}Y}$
with the property that $\rho'_{X_\mathrm{C}Y}$ is classical on $\cH_{X_\mathrm{C}}$ and $\rho'_{X_\mathrm{C}}=\rho_{X_\mathrm{C}}$.
Furthermore, the DPIs for $\cE'$ imply those for $\cE$.
\end{lemma}

\begin{proof}
  The first part of the lemma follows because classical information
  can be copied, and hence $\cE'$ can be decomposed into first copying
  $X_\mathrm{C}$, and then performing $\cE$\footnote{Alternatively, we can
  think of $\cE$ as the concatenation of $\cE'$ with a partial trace;
  this allows us to use the same output state $\rho'$ for both maps in
  the argument below.}.

  Suppose
  $\cE:\rho_{X_\mathrm{C}X_\mathrm{Q}}\mapsto\rho'_Y$. The
  second part follows because if
  ${I( X_\mathrm{C} X_\mathrm{Q} X_\mathrm{S} : X_\mathrm{T})_{\rho}}
  \geq {I(Y X_\mathrm{S} : X_\mathrm{T})_{\rho'}}$ is a valid DPI for
  $\cE$ then
  ${I( X_\mathrm{C} X_\mathrm{Q} X_\mathrm{S} : X_\mathrm{T})_{\rho}
    \geq I( X_\mathrm{C} Y X_\mathrm{S} : X_\mathrm{T})_{\rho'}}$ is
  valid for $\cE'$.  The second of these implies the first by the
  submodularity relation
  ${I( X_\mathrm{C} Y X_\mathrm{S} : X_\mathrm{T})_{\rho'}} \geq {I(Y
    X_\mathrm{S} : X_\mathrm{T})_{\rho'}}$.
\end{proof}

All the above (in)equalities are necessary conditions for a vector to
be an entropy vector compatible with the causal structure
$C^{\qQ}$. They constrain a polyhedral cone in
$\mathbb{R}_{\geq 0}^{m}$, where $m$ is the total number of coexisting
sets of $C^{\qQ}$,
$$\Gamma\left(C^{\qQ}\right):=\left\{v \in \mathbb{R}_{\geq 0}^{m} \mid
  M_\mathrm{B}\left(C^{\qQ}\right) \cdot v \geq 0,
  M_\mathrm{QCI}\left(C^{\qQ}\right) \cdot v=0,\text{ and }
  M_\mathrm{DPI}\left(C^{\qQ}\right) \cdot v \geq 0 \right\}.$$

\begin{example}[Entropic constraints for the quantum instrumental scenario]
  The cone
  ${\Gamma\left(IC^{\qQ}\right)=\left\{v \in \mathbb{R}_{\geq 0}^{15}
      \mid M_\mathrm{B}\left(IC^{\qQ}\right) \cdot v \geq 0,
      M_\mathrm{QCI}\left(IC^{\qQ}\right) \cdot v=0 \text{ and }
      M_\mathrm{DPI}\left(IC^{\qQ}\right) \cdot v \geq 0 \right\}}$
  involves the matrix $M_\mathrm{B}\left(IC^{\qQ}\right)$ that
  features $29$ (independent) inequalities\footnote{Note that the only weak
  monotonicity relations that are not made redundant by other basic
  inequalities are $H(A_Y|A_ZX)+H(A_Y)\geq 0$,
  $H(A_Z|A_YX)+H(A_Z)\geq 0$, $H(A_Y \mid A_Z)+H(A_Y \mid X) \geq 0$
  and $H(A_Z \mid A_Y)+H(A_Z \mid X) \geq 0$.}.  In this case a single
  independence constraint encodes that $X$ is independent of
  $A_YA_Z$:
  $$M_\mathrm{QCI}\left(IC^{\qQ}\right)=
  \left( \begin{array}{ccccccccccccccc} 0 & 0 & -1 & 0 & 0 & -1 & 0 &
      0 & 0 & 0 & 0 & 0 & 1 & 0 & 0\end{array} \right)\, .$$ Two data
  processing inequalities are required (cf.\
  Lemma~\ref{lemma:simplify_DPI_classicalquantum}),
  $I(A_Z X:A_Y) \geq I(XZ:A_Y)$ and $I(A_YZ:X) \geq {I(YZ:X)}$, which
  yield a matrix
  $$M_\mathrm{DPI}\left(IC^{\qQ}\right)=\left( \begin{array}{ccccccccccccccc}
  0 & 0 & 0 & 0 & 0 & 0 & 0 & 0 & 1 & 0 & -1 & 0 & -1 & 1 & 0 \\
  0 & 0 & 0 & 0 & 0 & 0 & 0 & 1 & 0 & 0 & 0 & -1 & 0 & -1 & 1 
  \end{array} \right).$$
The above matrices are all expressed in terms of coefficients of
$(H(A_Y), H(A_Z), H(X), H(Y), H(Z), H(A_YA_Z), \\ H(A_YX), H(A_YZ),
H(A_ZX),H(XY), H(XZ), H(YZ), H(A_YA_ZX), H(A_YXZ), H(XYZ))$.
Although the notation suppresses the different states there is no
ambiguity because e.g. the entropy of $X$ is the same for all states
with subsystem $X$. The full list of inequalities is provided in the
appendix.
\end{example}

From $\Gamma\left(C^{\qQ}\right)$, an outer approximation to the set
of compatible entropy vectors
$\Gamma^{*}_{\mathcal{M}}\left( C^{\qQ}\right)$ of the observed
scenario $\mathcal{M}$ of $C^{\qQ}$ can be obtained using
Fourier-Motzkin elimination. This leads to
$\Gamma_{\mathcal{M}}\left(C^{\qQ}\right):=\left\{w \in
  \mathbb{R}^{2^k-1}_{\geq 0} \mid \exists v \in
  \Gamma\left(C^{\qQ}\right) \mathrm{ s.t. }  w= \pi_\mathcal{M}(v)
\right\}$, which can be written as
$\Gamma_{\mathcal{M}}\left(C^{\qQ}\right)=\left\{w \in
  \Gamma_\mathcal{M} \mid M_\mathcal{M}\left(C^{\qQ}\right) \cdot w
  \geq 0 \right\}$.  The matrix $M_\mathcal{M}\left(C^{\qQ}\right)$
encodes all (in)equalities on $\mathbb{R}^{2^k-1}_{\geq 0}$ implied by
$M_\mathrm{B}\left(C^{\qQ}\right) \cdot v \geq 0$,
$M_\mathrm{QCI}\left(C^{\qQ}\right) \cdot v=0$ and
$M_\mathrm{DPI}\left(C^{\qQ}\right) \cdot v \geq 0$ (except for the
Shannon inequalities, which are already included in
$\Gamma_\mathcal{M}$). Note that
$\Gamma_{\mathcal{M}}\left(C^{\cC}\right)
\subseteq\Gamma_{\mathcal{M}}\left(C^{\qQ}\right) \subseteq
\Gamma_{\mathcal{M}}$, where the first relation holds because all
inequalities relevant for quantum states hold in the classical case as
well\footnote{This can be seen by thinking of a classical source as made up
of two or more (perfectly correlated) random variables as its
subsystems, which are sent to its children and processed there. The
Shannon inequalities hold among all of these variables (and also imply
any weak monotonicity constraints). The classical independence
relations include the quantum ones but may add constraints that
involve conditioning on any of the variables' ancestors. These
(in)equalities are tighter than the DPIs, which are hence not
explicitly considered in the classical case.}.

\begin{example}[Entropic outer approximation for the quantum instrumental scenario]
  The projection of $\Gamma\left(IC^{\qQ}\right)$ leads to the
  entropic cone
  $\Gamma_{\cM}\left(IC^{\qQ}\right)=\left\{w \in \Gamma_3 \mid
    M_\mathcal{M}\left(IC^{\qQ}\right) \cdot w \geq 0 \right\}$, for
  which $M_\mathcal{M}\left(IC^{\qQ}\right)$ equals
  $M_\mathcal{M}\left(IC^{\cC}\right)$ from
  Example~\ref{example:IC_class_obs}, thus corresponding to the
  constraint ${I(X:YZ)} \leq H(Z)$. Hence,
  $\overline{\Gamma^{*}_\cM}\left(IC^{\qQ}\right)$ coincides with
  $\overline{\Gamma^{*}_\cM}\left(IC^{\cC}\right)$~\cite{non_shan}.
\end{example}

This method has been applied to find an outer approximation to the
entropy cone of the triangle causal structure in the quantum case
(cf.\ Figure~\ref{fig:instrumental}(c))~\cite{Chaves2015}. This
approximation did not coincide with the outer approximation to the
classical triangle scenario obtained from Shannon inequalities and
independence constraints. Whether there are more as yet unknown
inequalities in the quantum case remains an open question (as opposed
to the classical case where even better outer approximations have
already been found~\cite{non_shan}). In~\cite{Chaves2015}, the method
was furthermore combined with the approach reviewed in
Section~\ref{sec:quantum_conditioning}, where
it was applied to a scenario related to $IC$ (cf.\
Example~\ref{example:information_causality} below).

\subsubsection{Causal structures with unobserved systems in other non-signalling theories} \label{sec:gDAGs}
The concept of a generalised causal structure was introduced
in~\cite{Henson2014}, the idea being to have one framework in which
classical, quantum and even more general systems, for instance
non-local boxes~\cite{Tsirelson1993,Popescu1994FP}, can be shared by
unobserved nodes and where theory independent features of networks and
corresponding bounds on our observations may be identified.
\begin{definition}
  A \emph{generalised causal structure} $C^{\gG}$ is a causal
  structure which for each observed node has an associated random
  variable and for each unobserved node has a corresponding
  non-signalling resource allowed by a generalised probabilistic
  theory.
\end{definition}
Classical and quantum causal structures can be viewed as special cases
of generalised causal structures~\cite{Henson2014, Fritz2015}.
Generalised probabilistic theories may be conveniently described in
the operational-probabilistic framework
of~\cite{Chiribella2010}. Circuit elements correspond to so-called
tests that are connected by wires, which represent propagating
systems. In general, such a test has an input system, and two outputs:
an output system and an outcome. In the case of a system with trivial
input this corresponds to a preparation test, and in case of trivial output
this is an observation-test.  In the causal structure framework, a
test is associated to each node. However, each such test has only one
output: for unobserved nodes this is a general resource state; for
observed nodes it is a random variable.  Furthermore, resource states
do not allow for signalling from the future to the past, i.e., we are
considering so-called causal operational-probabilistic theories. This
is important for the interpretation of generalised causal structures.

A distribution $P$ over the observed nodes of a generalised causal
structure $C^{\gG}$ is compatible with $C^{\gG}$ if there exists a
causal operational-probabilistic theory, a resource for each
unobserved edge in that theory and transformations for each node that
allow for the generation of $P$. We denote the set of all compatible
distributions $\mathcal{P}_{\mathcal{M}}\left( C^{\gG} \right)$.  As
in the quantum case, there is no notion of a joint state of all nodes
in the causal structure and of conditioning on an unobserved
system. Even more, there is no consensus on the representation of
states and their dynamics in general non-signalling theories. To
circumvent this, the classical notion of d-separation has been
reformulated~\cite{Henson2014}, which enables the following
proposition.
\begin{prop}[Henson, Lal \& Pusey] \label{prop:dseparationgDAG} Let
  $C^{\gG}$ be a generalised causal structure and let $X$, $Y$ and $Z$
  be pairwise disjoint subsets of observed nodes in $C^{\gG}$.  If a
  probability distribution $P$ is compatible with $C^{\gG}$, then the
  d-separation of $X$ and $Y$ by $Z$ implies the conditional
  independence $X \indep Y | Z$.  Conversely, if for every
  distribution $P$ compatible with $C^{\gG}$ the conditional
  independence $X \indep Y | Z$ holds, then $X$ is d-separated from
  $Y$ by $Z$ in $C^{\gG}$.
\end{prop}

This allows for the derivation of conditional independence relations
among observed variables that hold in any generalised probabilistic
theory, which hence restrict a general entropic cone.  Furthermore, it
rigorously justifies retaining the independence constraints among the
(observed) variables in coexisting sets in quantum causal structures
(cf.\ Section~\ref{sec:quantumlatent}), which can be seen as special
cases of generalised causal structures.

In~\cite{Henson2014}, sufficient conditions for identifying causal
structures $C$ for which in the classical case, $C^{\cC}$, there are
no restrictions on the distribution over observed variables other than
those that follow from the d-separation of these variables were
derived.  Since, by Proposition~\ref{prop:dseparationgDAG}, these
conditions also hold in $C^\qQ$ and $C^\gG$, this implies
$\cP_{\cM}(C^\cC)=\cP_{\cM}(C^\qQ)=\cP_{\cM}(C^\gG)$.  For causal
structures with up to six nodes, there are 21 cases (and some that can
be reduced to the these 21) where such equivalence does not hold and
where further relations among the observed variables have to be taken
into account~\cite{Henson2014, Pienaar2016}.

Outer approximations to the entropic cones for causal structures,
$C^{\gG}$, based on the observed variables and their independences
only were derived in~\cite{Henson2014}.  Moreover, a few additional
constraints for certain generalised causal structures were derived
there.  For example, the entropic constraint $I(X:Y)+I(X:Z) \leq H(X)$
for the triangle causal structure of Figure~\ref{fig:instrumental}(c)
(which had previously been established in the classical
case~\cite{Fritz2012}) was found. This constraint does not follow from
the observed independences, but nonetheless holds for the triangle
causal structure in generalised probabilistic theories.

In spite of this, a systematic entropic procedure, in which the
unobserved variables are explicitly modelled and then eliminated from
the description, is not available for generalised causal
structures. The issue is that we are lacking a generalisation of the
Shannon and von Neumann entropy to generalised probabilistic theories
that obeys submodularity and for which the conditional entropy can be
written as the difference of unconditional entropies~\cite{Short2010,
  Barnum2012}.

One possible generalised entropy is the \emph{measurement entropy},
which is positive and obeys some of the submodularity constraints
(those with $X_S\cap X_T=\{\}$) but not all~\cite{Short2010,
  Barnum2012}.  Using this, Ref.~\cite{Cadney2012a} considered the set
of possible entropy vectors for a bipartite state in \emph{box world},
a generalised probabilistic theory that permits all bipartite
correlations that are non-signalling~\cite{Barrett07}. They found no
further constraints on the set of possible entropy vectors in this
setting (hence, contrary to the quantum case, measurement entropy
vectors of separable states in box world can violate
monotonicity). Other generalised probabilistic theories and
multi-party states have, to our knowledge, not been similarly
analysed.

\subsubsection{Other directions for exploring quantum and generalised
  causal structures} \label{sec:other_stuff}
The approaches to quantum and generalised causal structures above are
based on adaptations of the theory of Bayesian networks to the
respective settings and on retaining the features that remain valid,
for instance the relation between d-separation and independence for
observed variables~\cite{Henson2014} (cf.\
Section~\ref{sec:gDAGs}). Other
approaches to generalise classical networks to the quantum realm have
been pursued~\cite{Leifer2013}, where a definition of conditional
quantum states, analogous to conditional probability distributions was
formulated.

Recent articles have proposed generalisations of Reichenbach's
principle~\cite{Reichenbach1956} to the quantum
realm~\cite{Pienaar2015, Costa2015, Allen2016}. In
Ref.~\cite{Pienaar2015} a graph separation rule, q-separation, was
introduced, whereas~\cite{Costa2015, Allen2016} rely on a formulation
of quantum networks in terms of quantum channels and their Choi
states.

An active area of research is the exploration of frameworks that allow
for indefinite causal structures~\cite{Hardy2005, Hardy2007,
  Hardy2009}. There are several approaches achieving this, such as the
process matrix formalism~\cite{Oreshkov2012}, which has lead to the
derivation of so called causal inequalities and the identification of
signalling correlations that are achievable in this framework, however
not with any predefined causal
structure~\cite{Oreshkov2012,Baumeler2014}.  Another framework that is
able to describe such scenarios is the theory of quantum
combs~\cite{Chiribella2013}, illustrated by a quantum switch, a
quantum bit controlling the circuit structure in a quantum
computation.  A recent framework with the aim to model cryptographic
protocols is also available~\cite{Portmann2017}.  Some initial results
on the analysis of indefinite causal structures with entropy have
recently appeared~\cite{Miklin17}.

In the classical, quantum and generalised causal structures considered
above only the observed classical information can be transmitted via a
link between two observed variables and, in particular, no additional
unobserved system. This understanding of the causal links encodes a
Markov condition.  In other situations, it can be convenient for the
links in the graph to represent a notion of future instead of direct
causation, see e.g.~\cite{Colbeck2013,Colbeck2016}.

\section{Entropy vector approach with
  post-selection} \label{sec:post-selection}
A technique that leads to additional, more fine-grained inequalities
is based on post-selecting on the values of parentless classical
variables.  This technique was pioneered by Braunstein and
Caves~\cite{Braunstein1988} and has been used to systematically derive
numerous entropic inequalities~\cite{Chaves2012, Fritz2013,
  Chaves2013, Pienaar2016, Chaves2016}.

\subsection{Post-selection in classical causal structures}\label{sec:classical_conditioning}
In the following we denote a random variable $X$ post-selected on the
event of another random variable, $Y$, taking a particular value,
$Y=y$, as $X_{\mid Y=y}$. The same notation is used for a set of
random variables $S=\left\{X_1,~X_2,~\ldots~,~X_n \right\}$, whose
joint distribution is conditioned on ${Y=y}$,
$S_{\rm \mid Y=y}=\left\{X_{1 {\rm \mid Y=y}},~X_{2 {\rm \mid
      Y=y}},~\ldots~,~X_{ n {\rm \mid Y=y}} \right\}$. The following
lemma can be understood as a generalisation of (a part of) Fine's
theorem~\cite{Fine1982, Fine1982a}.

\begin{lemma}\label{lemma:justify_conditioning}
  Let $C^{\cC}$ be a classical causal structure with a parentless
  observed node $X$ that takes values $X=1,2,\ldots,n$ and let $P$ be
  a joint distribution over all random variables
  $\Omega=X \cup X^{\uparrow} \cup X^{\nuparrow}$ in $C^{\cC}$ (with $P$
  compatible with $C^{\cC}$).  Then there exists a joint distribution
  $Q$ over the
  $n \cdot \left|X^{\uparrow}\right| + \left|X^{\nuparrow}\right|$
  random variables
  $\Omega_{\rm \mid X} := X^{\uparrow}_{\rm \mid X=1} \cup
  X^{\uparrow}_{\rm \mid X=2} \cup \cdots \cup X^{\uparrow}_{\rm \mid
    X=n} \cup X^{\nuparrow}$ such that
  $Q( X^{\uparrow}_{\rm \mid X=x} X^{\nuparrow})=P( X^{\uparrow}
  X^{\nuparrow} \mid X=x)$ for all $x\in\{1,\ldots,n\}$.
\end{lemma}
\begin{proof}
  The joint distribution over the random variables
  $X^{\uparrow} \cup X^{\nuparrow}$ in $C^{\cC}$ can be written as
  ${P(X^{\uparrow} X^{\nuparrow})= \sum_{x=1}^{n} P( X^{\uparrow} \mid
    X^{\nuparrow} X=x) P( X^{\nuparrow}) P(X=x).}$ Now take
  $Q( X^{\uparrow}_{\rm \mid X=1}~\cdots~ X^{\uparrow}_{\rm \mid X=n}
  X^{\nuparrow})= \prod_{x=1}^{n} {P( X^{\uparrow}\mid
    X^{\nuparrow}X=x)}P( X^{\nuparrow})$.  As required, this
  distribution has marginals
  $Q( X^{\uparrow}_{\rm \mid X=x} X^{\nuparrow})={P( X^{\uparrow} \mid
    X^{\nuparrow} X=x)} P( X^{\nuparrow})$.
\end{proof}

It is perhaps easiest to think about this lemma in terms of a new
causal structure $C_X^{\cC}$ on $\Omega_{\rm \mid X}$ that is related
to the original.  Roughly speaking the new causal structure is formed
by removing $X$ and replacing the descendants of $X$ with several
copies each of which have the same causal relations as in the original
causal structure (with no mixing between copies).  More precisely, if
$X$ is a parentless node in $C^{\cC}$ we can form a
\emph{post-selected causal structure} on $\Omega_{\rm \mid X}$
(post-selecting on $X$) as follows: {\bf (1)} For each pair of nodes
$A$, $B \in X^{\nuparrow}$ in $C^{\cC}$, make $A$ a parent of $B$ in
$C_X^{\cC}$ iff $A$ is a parent of $B$ in $C^{\cC}$. {\bf (2)} For
each node $B \in X^{\nuparrow}$ in $C^{\cC}$ and for each node
$A_{\mid X=x}$, make $B$ a parent of $A_{\mid X=x}$ in $C_X^{\cC}$ iff
$B$ is a parent of $A$ in $C^{\cC}$. {\bf (3)} For each pair of nodes,
$A_{\mid X=x}$ and $B_{\mid X=x}$, make $B_{\mid X=x}$ a parent of
$A_{\mid X=x}$ in $C_X^{\cC}$ iff $B$ is a parent of $A$ in
$C^{\cC}$. (Note that there is no mixing between different values of
$X=x$.) See Figures~\ref{fig:instrumental_conditioned}
and~\ref{fig:Pienaar_examples} and Example~\ref{example:conditional}
for illustrations.  This view gives us the following corollary of
Lemma~\ref{lemma:justify_conditioning}, which is an alternative
generalisation of Fine's theorem
\begin{lemma}\label{lemma:altfine}
  Let $C^{\cC}$ be a classical causal structure with a parentless
  observed node $X$ that takes values $X=1,2,\ldots,n$ and let $P$ be
  a joint distribution over all random variables
  $X \cup X^{\uparrow} \cup X^{\nuparrow}$ in $C^{\cC}$ (with $P$
  compatible with $C^{\cC}$). Then there exists a joint distribution
  $Q$ compatible with the post-selected causal structure $C^{\cC}_X$
  such that
  $Q( X^{\uparrow}_{\rm \mid X=x} X^{\nuparrow})={P( X^{\uparrow}
  X^{\nuparrow} \mid X=x)}$ for all $x\in\{1,\ldots,n\}$.
\end{lemma}

The distributions that are of interest in this new causal structure
are the marginals $Q(X^{\uparrow}_{\rm \mid X=x} X^{\nuparrow})$ for
all $x$ (and their interrelations), as they correspond to
distributions in the original scenario.  Any constraints on these
distributions derived in the post-selected scenario are by
construction valid for the (post-selected) distributions compatible
with the original causal structure.
\begin{example}[Post-selection in the instrumental
  scenario] \label{example:conditional} Consider the causal structure
  $IC$ where the parentless variable $X$ takes values $0$ or $1$. 
  For any $P$ compatible with $IC^{\cC}$, there exists a distribution
  $Q$ compatible with the post-selected causal structure
  (Figure~\ref{fig:instrumental_conditioned}(a)) such that
  ${Q(Z_{\mid X=0} Y_{\mid X=0} A)}={P(Z Y \mid A X=0)}P(A)$ and
  ${Q(Z_{\mid X=1} Y_{\mid X=1} A)}={P(Z Y \mid A X=1)}P(A)$. These
  marginals and their relations are of interest for the original
  scenario.
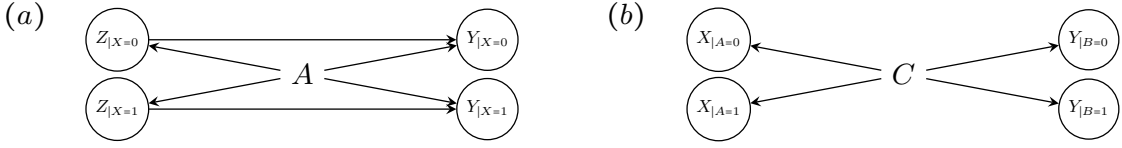
\begin{figure}
\centering 
\resizebox{0.85\columnwidth}{!}{
\begin{tikzpicture} [scale=0.5]
\node (1) at (-8,0.25){$(a)$};
\node[draw=black,circle,scale=0.6] (2) at (-6,-0.25) {$Z_{\mid X=0}$};
\node[draw=black,circle,scale=0.6] (2b) at (-6,-1.75) {$Z_{\mid X=1}$};
\node[draw=black,circle,scale=0.6] (3) at (2,-0.25) {$Y_{\mid X=0}$};
\node[draw=black,circle,scale=0.6] (3b) at (2,-1.75) {$Y_{\mid X=1}$};
\node (4) at (-2,-1.0) {$A$};

\draw [->,>=stealth] (2)--(3);
\draw [->,>=stealth] (4)--(2);
\draw [->,>=stealth] (4)--(3);
\draw [->,>=stealth] (2b)--(3b);
\draw [->,>=stealth] (4)--(2b);
\draw [->,>=stealth] (4)--(3b);

\node (M1) at (5,0.25){$(b)$};
\node[draw=black,circle,scale=0.6] (M2) at (7,-0.25) {$X_{\mid A=0}$};
\node[draw=black,circle,scale=0.6] (M2b) at (7,-1.75) {$X_{\mid A=1}$};
\node[draw=black,circle,scale=0.6] (M3) at (15,-0.25) {$Y_{\mid B=0}$};
\node[draw=black,circle,scale=0.6] (M3b) at (15,-1.75) {$Y_{\mid B=1}$};
\node (M4) at (11,-1.0) {$C$};

\draw [->,>=stealth] (M4)--(M2);
\draw [->,>=stealth] (M4)--(M3);
\draw [->,>=stealth] (M4)--(M2b);
\draw [->,>=stealth] (M4)--(M3b);
\end{tikzpicture}
}
\caption{(a) Pearl's instrumental scenario post-selected on binary
  $X$. The causal structure is obtained from the $IC$ by removing $X$
  and replacing $Y$ and $Z$ with copies, each of which has
  the same causal relations as in the original causal structure. (b)
  post-selected Bell scenario with binary inputs $A$ and $B$.}
\label{fig:instrumental_conditioned}
\end{figure}
\end{example}

Note that the above reasoning may be applied recursively. Indeed, the
causal structure with variables $\Omega_{\rm \mid X}$ may be
post-selected on the values of one of its parentless nodes.  The joint
distributions of the nodes $\Omega_{\mid X}$ and the associated causal
structure may be analysed in terms of entropies, as illustrated with
the following example.

\begin{example}[Entropic constraints for the post-selected Bell
  scenario~\cite{Braunstein1988}]\label{example:bell}
  In the Bell scenario with binary inputs $A$ and $B$
  (Figure~\ref{fig:instrumental}(b)), Lemma~\ref{lemma:altfine} may be
  applied first to post-select on the values of $A$ and then of
  $B$. This leads to a distribution $Q$ compatible with the
  post-selected causal structure (on $A$ and $B$) shown in
  Figure~\ref{fig:instrumental_conditioned}(b), for which
  $Q(X_{\rm \mid A=a} Y_{\rm \mid B=b})=P(X Y \mid A=a, B=b)$ for
  $a, b \in \left\{0,1\right\}$\footnote{In this case the joint
    distribution is already known to exist by Fine's
    theorem~\cite{Fine1982, Fine1982a}.}.  Applying the entropy vector
  method to the post-selected causal structure and marginalising to
  vectors of form
  $(H(X_{\rm \mid A=0}),H(X_{\rm \mid A=1}),H(Y_{\rm \mid
    B=0}),H(Y_{\rm \mid B=1}),H(X_{\rm \mid A=0}Y_{\rm \mid
    B=0}),H(X_{\rm \mid A=0}Y_{\rm \mid B=1}), H(X_{\rm \mid
    A=1}Y_{\rm \mid B=0}), H(X_{\rm \mid A=1}Y_{\rm \mid B=1}))$
  yields the inequality
  $H(Y_1|X_1)+H(X_1|Y_0)+H(X_0|Y_1)-H(X_0|Y_0) \geq 0$ and its
  permutations~\cite{Braunstein1988, Chaves2012}.\footnote{Whenever
    the input nodes take more than two values, the latter may be
    partitioned into two sets, guaranteeing applicability of these
    inequalities. Furthermore,~\cite{Chaves2013} showed that these
    inequalities are sufficient for detecting any behaviour that is
    not classically reproducible in the Bell scenario where the two
    parties perform measurements with binary outputs.}
\end{example}

The extension of Fine's theorem to more general Bell
scenarios~\cite{Liang2011, Abramsky2011}, i.e., to scenarios involving
a number of spacelike separated parties that each choose input values
and produce some output random variable (and scenarios that can be
reduced to the latter), has been combined with the entropy vector
method in~\cite{Chaves2012, Fritz2013}.

Entropic constraints that are derived in this way provide novel and
non-trivial entropic inequalities for the distributions compatible
with the original classical causal structure. Ref.~\cite{Fritz2013}
introduced this idea and analysed the so-called $n$-cycle scenario,
which is of particular interest in the context of non-contextuality
and includes the Bell scenario (with binary inputs and outputs) as a
special case\footnote{A full probabilistic characterisation of the $n$-cycle
scenario was given in~\cite{Araujo2013}.}.

In Ref.~\cite{Chaves2012}, new entropic inequalities for the
bilocality scenario, which is relevant for entanglement
swapping~\cite{Branciard2010, Branciard2012}, as well as quantum
violations of the classical constraints on the 4- and 5-cycle
scenarios were derived.  For the $n$-cycle scenario, the (polynomial
number of) entropic inequalities are sufficient for the detection of
any non-local distribution~\cite{Chaves2013}\footnote{This is also
  true of the exponential number of inequalities in the probabilistic
  case~\cite{Araujo2013}.}.  In the following we illustrate the method
of~\cite{Chaves2012, Fritz2013} with a continuation of
Example~\ref{example:conditional}.

\begin{example}[Entropic approximation for the post-selected instrumental scenario]
  The entropy vector method from Section~\ref{sec:entropicappr} is
  applied to the $5$-variable causal structure of
  Figure~\ref{fig:instrumental_conditioned}(a). The marginalisation is
  performed to retain all marginals that correspond to distributions
  in the original causal structure (Figure~\ref{fig:instrumental}(a)),
  i.e., any marginals of ${P(Y Z \mid X=0)}$ and ${P(Y Z \mid X=1)}$.
  Hence, the $5$ variable entropic cone is projected to a cone that
  restricts vectors of the form
  ${(H(Y_{\rm \mid X=0}),~H(Y_{\rm \mid X=1}),~H(Z_{\rm \mid
      X=0}),~H(Z_{\rm \mid X=1}),~H(Y_{\rm \mid X=0}Z_{\rm \mid
      X=0}),~H(Y_{\rm \mid X=1}Z_{\rm \mid X=1}))}$.  Note that
  entropies of unobserved marginals such as
  $H(Y_{\rm \mid X=0}Z_{\rm \mid X=1})$ are not included.  With this
  technique, the Shannon constraints for the three components
  ${(H(Y_{\rm \mid X=0}),~H(Z_{\rm \mid X=0}),~H(Y_{\rm \mid
      X=0}Z_{\rm \mid X=0}))}$ are recovered (the same holds for
  $X=1$); no additional constraints arise here.

  It is interesting to compare this to the Bell scenario considered in
  Example~\ref{example:bell}. In both causal structures any
  $4$-variable distributions,
  $P_{Z_{\rm \mid X=0}Z_{\rm \mid X=1}Y_{\rm \mid X=0}Y_{\rm \mid
      X=1}}$
  and
  $P_{X_{\rm \mid A=0}X_{\rm \mid A=1}Y_{\rm \mid B=0}Y_{\rm \mid
      B=1}}$
  respectively, are achievable\footnote{The additional causal links
    in Figure~\ref{fig:instrumental_conditioned}(b) do not affect the
    set of compatible distributions.}. However, the marginal entropy
  vector in the Bell scenario has more components, leading to
  additional constraints on the observed
  variables~\cite{Braunstein1988, Chaves2012}.
\end{example}

In some cases two different causal structures, $C_1$ and $C_2$, can
yield the same set of distributions after marginalising, a fact that
has been further explored in~\cite{Budroni2016b}. When this occurs,
either causal structure can be imposed when identifying the set of
achievable marginal distributions in either scenario. If the
constraints implied by the causal structure $C_1$ are a subset of
those implied by $C_2$, then those of $C_2$ can be used to compute
improved outer approximations on the entropic cone for $C_1$.
Furthermore, valid independence constraints may speed up computations
even if they do not lead to any new relations for the observed
variables\footnote{Note that some care has to be taken when identifying valid
constraints for scenarios with causal structure~\cite{Budroni2016b}.}.
Similar considerations also yield a criterion for indistinguishability
of causal structures in certain marginal scenarios --- if $C_1$ and
$C_2$ yield the same set of distributions after marginalising then
they cannot be distinguished in that marginal scenario.

In examples like the above, where no new constraints follow from
post-selection, it may be possible to introduce additional input
variables in order to certify the presence of quantum nodes in a
network. The new parentless nodes can then be used to apply
Lemma~\ref{lemma:justify_conditioning} and the above entropic
techniques. Mathematically, introducing further nodes to a causal
structure is always possible. However, this is only interesting if
experimentally feasible, e.g.\ if an experimenter has control over
certain observed nodes and is able to devise an experiment where he
can change their inputs. In the instrumental scenario this may be of
interest.

\begin{example}[Variations of the instrumental scenario] 
  In this scenario (Figure~\ref{fig:instrumental}(a)), a measurement
  on system $A_Z$ is performed depending on $X$ (where in the
  classical case $A_Z$ can w.l.o.g. be taken to be a copy of the
  unobserved random variable $A$). Its outcome $Z$ (in the classical
  case a function of $A$) is used to choose another measurement to be
  performed on $A_Y$ to generate $Y$ (classically another a copy of
  $A$).  It may often be straightforward for an experimenter to choose
  between several measurements. In the causal structure this
  corresponds to introducing an additional observed input $S$ to the
  second measurement (with the values of $S$ corresponding to
  different measurements on $A_Y$). Such an adaptation is displayed in
  Figure~\ref{fig:information_causality_DAG}(a).\footnote{Note that for
  ternary $S$ the outer approximation of the post-selected causal
  structure of Figure~\ref{fig:information_causality_DAG}(d) with
  Shannon inequalities does not lead to any interesting constraints
  (as opposed to the structure of
  Figure~\ref{fig:information_causality_DAG}(e), which is analysed
  further in Example~\ref{example:information_causality}).}

  Alternatively, it may be possible that the first measurement (on
  $A_Z$) is chosen depending on a combination of different,
  independent factors, which each correspond to a random variable
  $X_i$. For two variables $X_1$ and $X_2$ the corresponding causal
  structure is displayed in
  Figure~\ref{fig:information_causality_DAG}(b)\footnote{This is an example of
  a causal structure where non-Shannon inequalities among classical
  variables lead to a strictly tighter outer approximation in the
  classical and quantum case than the approximations derived using
  only Shannon and weak-monotonicity constraints (also if there is a
  causal link from $X_1$ to $X_2$)~\cite{non_shan}.}.

  Taken together, these two adaptations yield the causal structure of
  Figure~\ref{fig:information_causality_DAG}(c), relevant in the
  context of the principle of information
  causality~\cite{Pawlowski2009} (see also
  Example~\ref{example:information_causality} below).
\end{example}
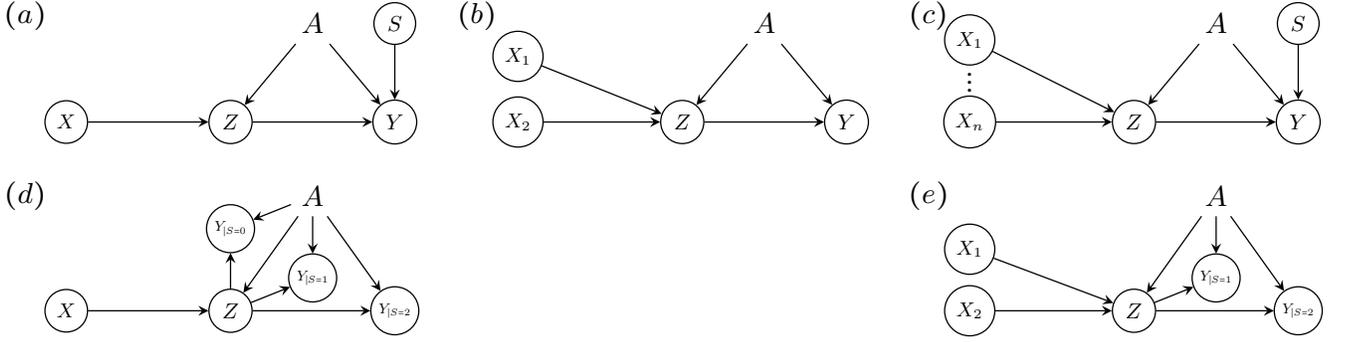
\begin{figure}
\centering 
\resizebox{1\columnwidth}{!}{
\begin{tikzpicture} [scale=0.9]
\node (B0) at (-13.5,1.3) {$(a)$};
\node[draw=black,circle,scale=0.7] (B1) at (-13,0) {$X$};
\node[draw=black,circle,scale=0.75] (B2) at (-11,0) {$Z$};
\node[draw=black,circle,scale=0.75] (B3) at (-9,0) {$Y$};
\node[draw=black,circle,scale=0.75] (B5) at (-9,1.2) {$S$};
\node (B4) at (-10,1.2) {$A$};
\draw [->,>=stealth] (B1)--(B2);
\draw [->,>=stealth] (B2)--(B3);
\draw [->,>=stealth] (B4)--(B2);
\draw [->,>=stealth] (B4)--(B3);
\draw [->,>=stealth] (B5)--(B3);

\node (A0) at (-8,1.3) {$(b)$};
\node[draw=black,circle,scale=0.7] (A1a) at (-7.5,0.8) {$X_1$};
\node[draw=black,circle,scale=0.7] (A1) at (-7.5,0) {$X_2$};
\node[draw=black,circle,scale=0.75] (A2) at (-5.5,0) {$Z$};
\node[draw=black,circle,scale=0.75] (A3) at (-3.5,0) {$Y$};
\node (A4) at (-4.5,1.2) {$A$};
\draw [->,>=stealth] (A1)--(A2);
\draw [->,>=stealth] (A1a)--(A2);
\draw [->,>=stealth] (A2)--(A3);
\draw [->,>=stealth] (A4)--(A2);
\draw [->,>=stealth] (A4)--(A3);

\node (0) at (-2.5,1.3) {$(c)$};
\node[draw=black,circle,scale=0.7] (1a) at (-2,1) {$X_1$};
\node[draw=black,circle,scale=0.7] (1) at (-2,0) {$X_n$};
\node(1b) at (-2,0.55) {$\cdot$};
\node(1c) at (-2,0.45) {$\cdot$};
\node(1d) at (-2,0.35) {$\cdot$};
\node[draw=black,circle,scale=0.75] (2) at (-0,0) {$Z$};
\node[draw=black,circle,scale=0.75] (3) at (2,0) {$Y$};
\node[draw=black,circle,scale=0.75] (5) at (2,1.2) {$S$};
\node (4) at (1,1.2) {$A$};
\draw [->,>=stealth] (1)--(2);
\draw [->,>=stealth] (1a)--(2);
\draw [->,>=stealth] (2)--(3);
\draw [->,>=stealth] (4)--(2);
\draw [->,>=stealth] (4)--(3);
\draw [->,>=stealth] (5)--(3);

\node (M) at (-2.5,-0.9) {$(e)$};
\node[draw=black,circle,scale=0.7] (M0) at (-2,-1.55) {$X_1$};
\node[draw=black,circle,scale=0.7] (M1) at (-2,-2.3) {$X_2$};
\node[draw=black,circle,scale=0.75] (M2) at (0,-2.3) {$Z$};
\node[draw=black,circle,scale=0.5] (M5) at (1,-1.9) {$Y_{\mid S=1}$};
\node[draw=black,circle,scale=0.5] (M3) at (2,-2.3) {$Y_{\mid S=2}$};
\node (M4) at (1,-0.9) {$A$};
\draw [->,>=stealth] (M0)--(M2);
\draw [->,>=stealth] (M1)--(M2);
\draw [->,>=stealth] (M2)--(M3);
\draw [->,>=stealth] (M2)--(M5);
\draw [->,>=stealth] (M4)--(M2);
\draw [->,>=stealth] (M4)--(M3);
\draw [->,>=stealth] (M4)--(M5);

\node (D0) at (-13.5,-0.9) {$(d)$};
\node[draw=black,circle,scale=0.7] (D1) at (-13,-2.3) {$X$};
\node[draw=black,circle,scale=0.75] (D2) at (-11,-2.3) {$Z$};
\node[draw=black,circle,scale=0.5] (D6) at (-11,-1.3) {$Y_{\mid S=0}$};
\node[draw=black,circle,scale=0.5] (D3) at (-9,-2.3) {$Y_{\mid S=2}$};
\node[draw=black,circle,scale=0.5] (D5) at (-10,-1.9) {$Y_{\mid S=1}$};
\node (D4) at (-10,-0.9) {$A$};
\draw [->,>=stealth] (D1)--(D2);
\draw [->,>=stealth] (D2)--(D3);
\draw [->,>=stealth] (D4)--(D2);
\draw [->,>=stealth] (D4)--(D3);
\draw [->,>=stealth] (D4)--(D5);
\draw [->,>=stealth] (D2)--(D5);
\draw [->,>=stealth] (D4)--(D6);
\draw [->,>=stealth] (D2)--(D6);
\end{tikzpicture}
}
\caption{Variations of the instrumental scenario (a), (b) and (c). The
  causal structure (c) is relevant for the derivation of the
  information causality inequality where $S$ takes $n$ possible
  values. (d) and (e) are the causal structures that are effectively
  analysed when post-selecting on a ternary $S$ in (a) and on a binary
  $S$ in (c) respectively.}
\label{fig:information_causality_DAG}
\end{figure}

A second approach that relies on very similar ideas (also justified by
Lemma~\ref{lemma:justify_conditioning}) is taken
in~\cite{Pienaar2016}. For a causal structure $C^{\cC}$ with nodes
$\Omega = X \cup X^{\uparrow} \cup X^{\nuparrow}$, where $X$ is a
parentless node, conditioning the joint distribution over all nodes on
a particular $X=x$ retains the independences of $C^{\cC}$.  In
particular, the conditioning does not affect the distribution of the
$X^{\nuparrow}$, i.e., ${P(X^{\nuparrow}\mid X=x)=P(X^{\nuparrow})}$
for all $x$.  The corresponding entropic constraints can be used to
derive entropic inequalities without the detour over computing large
entropic cones, which may be useful where the latter computations are
infeasible. The constraints that are used in~\cite{Pienaar2016} are,
however, a (diligently but somewhat arbitrarily chosen) subset of the
constraints that would go into the entropic technique detailed earlier
in this section for the full causal structure. Indeed, when the
computations are feasible, applying the full entropy vector method to
the corresponding post-selected causal structure gives a systematic
way to derive constraints, which are in general strictly tighter (cf.\
Example~\ref{example:pienaar_structure_analysis}).

So far, the restricted technique has been used in~\cite{Pienaar2016}
to derive the entropic inequality
\begin{equation}
  I(X_{\rm \mid C=0}:Z)-I(Y_{\rm \mid C=0}:Z)-I(X_{\rm \mid C=1}:Z)+I(Y_{\rm \mid C=1}:Z)\leq H(Z),\label{eq:pienaar_inequality}
\end{equation}
which is valid for all the classical causal structures of
Figure~\ref{fig:Pienaar_examples} (previously considered
in~\cite{Henson2014}). The inequality was used to certify the
existence of classical distributions that respect the conditional
independence constraints among the observed variables but that are not
achievable in the respective causal structures\footnote{These causal
  structures may thus also allow for quantum correlations that are not
classically achievable.}.  In the following we look at these three
causal structures in more detail and illustrate the relation between
the two techniques.
\begin{figure}
\centering
\resizebox{1\columnwidth}{!}{
\begin{tikzpicture}[scale=0.62]
\node (0) at (-2.75,3) {$(a)$};
\node[draw=black,circle,scale=0.75] (X) at (-2,2) {$X$};
\node[draw=black,circle,scale=0.75] (Y) at (2,2) {$Y$};
\node[draw=black,circle,scale=0.75] (Z) at (0,-1.46) {$Z$};
\node (A) at (1,0.28) {$A$};
\node (B) at (-1,0.28) {$B$};
\node[draw=black,circle,scale=0.75] (C) at (0,2) {$C$};

\draw [->,>=stealth] (A)--(Y); 
\draw [->,>=stealth] (A)--(Z); 
\draw [->,>=stealth] (B)--(X); 
\draw [->,>=stealth] (B)--(Z); 
\draw [->,>=stealth] (C)--(X); 
\draw [->,>=stealth] (C)--(Y); 

\node (M0) at (3.25,3) {$(b)$};
\node[draw=black,circle,scale=0.75] (MX) at (4,2) {$X$};
\node[draw=black,circle,scale=0.75] (MY) at (8,2) {$Y$};
\node[draw=black,circle,scale=0.75] (MZ) at (6,-1.46) {$Z$};
\node (MA) at (7,0.28) {$A$};
\node (MB) at (5,0.28) {$B$};
\node[draw=black,circle,scale=0.75] (MC) at (6,2.75) {$C$};

\draw [->,>=stealth] (MA)--(MY); 
\draw [->,>=stealth] (MA)--(MZ); 
\draw [->,>=stealth] (MB)--(MX); 
\draw [->,>=stealth] (MB)--(MZ); 
\draw [->,>=stealth] (MC)--(MX); 
\draw [->,>=stealth] (MC)--(MY); 
\draw [->,>=stealth] (MX)--(MY);

\node (N0) at (9.25,3) {$(c)$};
\node[draw=black,circle,scale=0.75] (NX) at (10,2) {$X$};
\node[draw=black,circle,scale=0.75] (NY) at (14,2) {$Y$};
\node[draw=black,circle,scale=0.75] (NZ) at (12,-1.46) {$Z$};
\node (NA) at (13,0.28) {$A$};
\node (NB) at (11,0.28) {$B$};
\node[draw=black,circle,scale=0.75] (NC) at (12,2.75) {$C$};
\draw [->,>=stealth] (NA)--(NY); 
\draw [->,>=stealth] (NA)--(NZ); 
\draw [->,>=stealth] (NB)--(NX); 
\draw [->,>=stealth] (NB)--(NZ); 
\draw [->,>=stealth] (NC)--(NX); 
\draw [->,>=stealth] (NX)--(NY);

\node (X0) at (15.25,3) {$(d)$};
\node[draw=black,circle,scale=0.75] (XX0) at (16,2) {$X0$};
\node[draw=black,circle,scale=0.75] (XX1) at (17,2) {$X1$};
\node[draw=black,circle,scale=0.75] (XY0) at (19,2) {$Y0$};
\node[draw=black,circle,scale=0.75] (XY1) at (20,2) {$Y1$};
\node[draw=black,circle,scale=0.75] (XZ) at (18,-1.5) {$Z$};
\node (XA) at (19,0.5) {$A$};
\node (XB) at (17,0.5) {$B$};

\draw [->,>=stealth] (XA)--(XY0); 
\draw [->,>=stealth] (XA)--(XY1);
\draw [->,>=stealth] (XA)--(XZ); 
\draw [->,>=stealth] (XB)--(XX0); 
\draw [->,>=stealth] (XB)--(XX1); 
\draw [->,>=stealth] (XB)--(XZ); 

\node (X0) at (21.25,3) {$(e)$};
\node[draw=black,circle,scale=0.75] (YX0) at (22.5,2.3) {$X0$};
\node[draw=black,circle,scale=0.75] (YX1) at (23.5,1.7) {$X1$};
\node[draw=black,circle,scale=0.75] (YY0) at (25.5,1.7) {$Y1$};
\node[draw=black,circle,scale=0.75] (YY1) at (26.5,2.3) {$Y0$};
\node[draw=black,circle,scale=0.75] (YZ) at (24.5,-1.5) {$Z$};
\node (YA) at (25.5,0.5) {$A$};
\node (YB) at (23.5,0.5) {$B$};

\draw [->,>=stealth] (YA)--(YY0); 
\draw [->,>=stealth] (YA)--(YY1);
\draw [->,>=stealth] (YA)--(YZ); 
\draw [->,>=stealth] (YB)--(YX0); 
\draw [->,>=stealth] (YB)--(YX1); 
\draw [->,>=stealth] (YB)--(YZ); 
\draw [->,>=stealth] (YX0)--(YY1);
\draw [->,>=stealth] (YX1)--(YY0);

\end{tikzpicture}
}
\caption{Causal structures from
  Example~\ref{example:pienaar_structure_analysis}. Post-selecting on
  a binary observed variable $C$ leads to the causal structure (d) in
  the case of structure (a), whereas both (b) and (c) lead to
  structure (e). In particular, this shows that the conditional
  techniques may yield the same results for different causal
  structures. }
\label{fig:Pienaar_examples}
\end{figure}
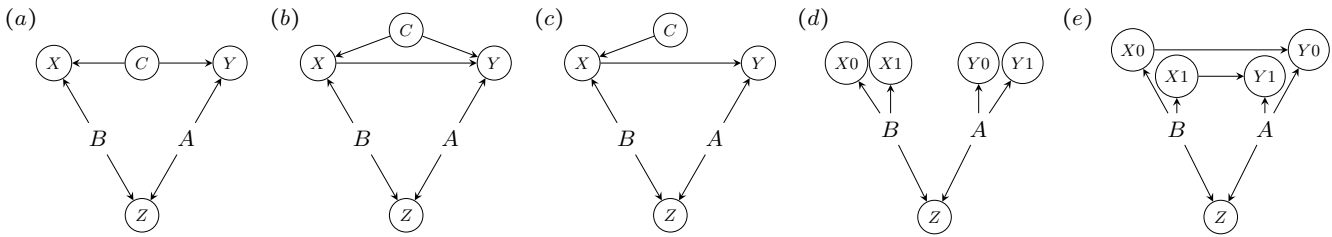
\begin{example}\label{example:pienaar_structure_analysis}
  Applying the post-selection technique for a binary random variable
  $C$ to the causal structure from
  Figure~\ref{fig:Pienaar_examples}(a) yields the effective causal
  structure \ref{fig:Pienaar_examples}(d). The latter can be analysed
  with the above entropy vector method, which leads to a cone that is
  characterised by $14$ extremal rays or equivalently in terms of $22$
  inequalities, both available in the appendix.  The
  inequalities $I(Z:X_{\rm \mid C=1})\geq 0$,
  $I(Z:Y_{\rm \mid C=0})\geq 0$,
  $I(X_{\rm \mid C=1}:Y_{\rm \mid C=1} \mid Z)\geq 0$ and
  $H(Z\mid X_{\rm \mid C=0}) \geq I(X_{\rm \mid C=1} Z:Y_{\rm \mid
    C=1})$, which are part of this description, imply
  \eqref{eq:pienaar_inequality} above. We are not aware of any quantum
  violations of these inequalities.
 
  Structures (b) and (c) both lead to the causal structure (e) upon
  post-selecting on a binary $C$. The latter causal structure turns
  out to be computationally harder to analyse with the entropy vector
  method and we have not been able to perform the corresponding
  marginalisation when taking all Shannon and independence constraints
  into account\footnote{We were working with conventional variable
    elimination software on a desktop computer}. Hence, the method
  outlined in~\cite{Pienaar2016} is a useful alternative here.
\end{example}

\subsection{Post-selection in quantum and general non-signalling causal structures} \label{sec:quantum_conditioning}
\label{sec:generalised_cone}
In causal structures with quantum and more general non-signalling
nodes, Lemma~\ref{lemma:justify_conditioning} is not valid.  For
instance, Bell's theorem can be recast as the statement that there are
distributions compatible with the quantum Bell scenario for which
there is no joint distribution of $X_{\rm \mid A=0}$,
$X_{\rm \mid A=1}$, $Y_{\rm \mid B=0}$ and $Y_{\rm \mid B=1}$ in the
post-selected causal structure (on $A$ and $B$) that has the required
marginals (in the sense of Lemma~\ref{lemma:altfine}).

Nonetheless, the post-selection technique has been generalised to such
scenarios~\cite{Chaves2015, Chaves2016}, i.e., it is still possible to
post-select on parentless observed (and therefore classical) nodes
taking specific values.  In such scenarios the observed variables can
be thought of as obtained from the unobserved resources by means of
measurements or tests. If a descendant of the variable that is
post-selected on has quantum or general non-signalling nodes as
parents, then the different instances of the latter node and of all
its descendants do not coexist (even if they are observed, hence
classical). This is because such observed variables are generated by
measuring a quantum or other non-signalling system. Such a system is
altered (or destroyed) in a measurement, hence does not allow for the
simultaneous generation of different instances of its children due to
the impossibility of cloning.

In the quantum case, this is reflected in the identification of the
coexisting sets in the post-selected causal
structure\footnote{Different instances of a variable after
  post-selection have to be seen as alternatives and not as
  simultaneous descendants of their parent node as the representation
  of the post-selected causal structure might suggest.}, as is
illustrated with the following example.

\begin{example}[Information causality scenario in the quantum case~\cite{Chaves2015}]\label{example:information_causality}
  The communication scenario used to derive the principle of
  information causality~\cite{Pawlowski2009} is based on the variation
  of the instrumental scenario displayed in
  Figure~\ref{fig:information_causality_DAG}(c). It has been analysed
  with the entropy vector method in Ref.~\cite{Chaves2015}, an
  analysis that is presented in the following.

  Conditioning on values of the variable $S$ is possible in the
  classical and quantum cases. However, whereas in the classical case
  the variables $Y_{\rm \mid S=s}$ for different $S$ share a joint
  distribution (cf.\ Lemma~\ref{lemma:justify_conditioning}), they do
  not coexist in the quantum case.  For binary $S$, the coexisting
  sets are $\left\{X_1,~X_2,~A_Z,~A_Y\right\}$,
  $\left\{X_1,~X_2,~Z,~A_Y\right\}$,
  $\left\{X_1,~X_2,~Z,~Y_{\rm \mid S=1}\right\}$ and
  $\left\{X_1,~X_2,~Z,~Y_{\rm \mid S=2}\right\}$. The only
  independence constraints in the quantum case are that $X_1$, $X_2$
  and $A_YA_Z$ are mutually independent. Marginalising until only
  entropies of $\{X_1, Y_{\rm \mid S=1}\}$,
  $\{X_2, Y_{\rm \mid S=2}\}$, $\{Z\}$ and their subsets remain,
  yields only one non-trivial inequality,
  ${\sum_{s=1}^{n} I(X_s : Y_{\rm \mid S=s}) \leq H(Z),}$ with
  $n=2$. \footnote{Note that this is slightly adapted
    from~\cite{Chaves2015} where they found
  $I(X_1 : Y_{\rm \mid S=1})+I(X_2 : Y_{\rm \mid S=2}) \leq
  H(Z)+I(X_1:X_2)$, as $X_1$ and $X_2$ were not assumed
  independent there. Furthermore, this is also the only inequality found in
  the classical case when restricting to this same marginal
  scenario~\cite{Chaves2016}.}  The same inequality was previously
  derived by Paw{\l}owski et al.\ for general $n$~\cite{Pawlowski2009},
  where the choice of marginals was inspired by the communication task
  considered. Subsequently, Ref.~\cite{Chaves2015} considered another
  marginal scenario --- the one with with coexisting sets
  $\{X_1,X_2,Z,Y_{\rm \mid S=1}\}$, $\{X_1,X_2,Z,Y_{\rm \mid S=2}\}$
  and all of their subsets --- which led to additional inequalities.
\end{example}

Similar considerations were applied to causal structures allowing for
general non-signalling resources, $C^{\gG}$ in~\cite{Chaves2016}. Let
$O=X_{\rm O}^{\uparrow} \cup X_{\rm O}^{\nuparrow} \cup X$ be the
disjoint union of its observed nodes, where $X_{\rm O}^{\uparrow}$ are
the observed descendants and $X_{\rm O}^{\nuparrow}$ the observed
non-descendants of $X$. If the variable $X$ takes values
$x \in \left\{1,~2,~\ldots,~n \right\}$, this leads to a joint
distribution of $X_{\rm O}^{\uparrow} \cup X_{\rm O}^{\nuparrow}$ for
each $X=x$, i.e., there is a joint distribution for
${P(X_{\rm O}^{\uparrow} X_{\rm O}^{\nuparrow} \mid X=x)=P(X_{\rm
    O}^{\uparrow} \mid X_{\rm O}^{\nuparrow} X=x)P(X_{\rm
    O}^{\nuparrow})}$
for all $x$, denoted
$P(X_{\rm O { \mid X=x}}^{\uparrow} X_{\rm O}^{\nuparrow})$.  Because
$X$ does not affect the distribution of the independent variables
$X_{\rm O}^{\nuparrow}$, the distributions
$P(X_{\rm O { \mid X=x}}^{\uparrow} X_{\rm O}^{\nuparrow})$ have
coinciding marginals on $X_{\rm O}^{\nuparrow}$, i.e.,
${P(X_{\rm O}^{\nuparrow})=\sum_sP(X_{\rm O { \mid X=x}}^{\uparrow}=s,
  X_{\rm O}^{\nuparrow})}$
for all $x$, where $s$ runs over the alphabet of
$X_{\rm O}^{\uparrow}$. This encodes no-signalling
constraints. \footnote{Note that there may be other constraints
  that arise from no-signalling. For instance
  Example~\ref{example:ns_ic} suggests further constraints for
  each $P(X_{\rm O { \mid X=x}}^{\uparrow} X_{\rm O}^{\nuparrow})$
  are implied by requiring non-signalling resources. The latter
  have to be found and added to the description separately.}

In terms of entropy, there are $n$ entropic cones, one for each
$P(X_{\rm O { \mid X=x}}^{\uparrow} X_{\rm O }^{\nuparrow})$ (which
each encode the independences among the observed variables). According
to the above, they are required to coincide on the entropies for
$X_{\rm O}^{\nuparrow}$ and on all of its subsets. These constraints
define a convex polyhedral cone that is an outer approximation to the
set of all entropy vectors achievable in the causal
structure. Whenever the distributions
$P(X_{\rm O { \mid X=x}}^{\uparrow} X_{\rm O }^{\nuparrow})$ involve
fewer than three variables and assuming that all constraints implied by
the causal structure and no-signalling have been taken into account\footnote{Note that it may not always be
obvious how to identify all relevant constraints (cf.\ the conjectured
constraints in Example~\ref{example:ns_ic}).},
this approximation is tight because
$\overline{\Gamma^{*}_3}=\Gamma_3$.

Several examples of the use of this technique can be found in
Ref.~\cite{Chaves2016}, including the original information causality
scenario (which we discuss in Example~\ref{example:ns_ic}) and
an entropic analogue of monogamy relations for Bell inequality
violations~\cite{Masanes2006,Pawlowski2009b}.

\begin{example}[Information causality scenario in general
  non-signalling theories]\label{example:ns_ic}
  This is related to Example~\ref{example:information_causality} above
  and reproduces an analysis from~\cite{Chaves2016}. In this marginal
  scenario we consider the Shannon cones for the three sets
  $\left\{X_1,Y_{\rm\mid S=1}\right\}$,
  $\left\{X_2,Y_{\rm\mid S=2}\right\}$ and $\left\{Z\right\}$ as well
  as the constraints $I(X_1:Y_{\rm\mid S=1})\leq H(Z)$ and
  $I(X_2:Y_{\rm\mid S=2})\leq H(Z)$ which are conjectured to
  hold~\cite{Chaves2016}. (This conjecture is based on an argument
  in~\cite{Popescu2014} that covers a special case; we are not aware
  of a general proof.)

  These conditions constrain a polyhedral cone of vectors
  $(H(X_1),~H(X_2),~H(Z),~H(Y_{\rm\mid S=1}),$ $H(Y_{\rm\mid S=2}),~H(X_1Y_{\rm\mid S=1}),~H(X_2Y_{\rm\mid S=2}))$, with
  $8$ extremal rays that are all achievable using
  PR-boxes~\cite{Tsirelson1993,Popescu1994FP}.  Importantly, the
  stronger constraint
  $I(X_1:Y_{\rm \mid S=1})+I(X_2:Y_{\rm \mid S=2})\leq H(Z)$, which
  holds in the quantum case, (cf.\
  Example~\ref{example:information_causality}) does not hold here.
\end{example}

\section{Alternative techniques} \label{sec:further_techniques}
Instead of relaxing the problem of characterising the set of
probability distributions compatible with a causal structure by
considering entropy vectors, other computational techniques are
currently being developed. In the following, we give a brief overview
of these methods.

In this context, note also that there are methods that allow
certification that the only restrictions implied by a causal structure
are the conditional independence constraints among the observed
variables~\cite{Henson2014}, as well as procedures to show that the
opposite is the case~\cite{Evans2012, Evans2015}.  Such methods may
(when applicable) indicate whether a causal structure should be
analysed further (corresponding techniques are reviewed
in~\cite{Pienaar2016}).

\subsection{Entropy vectors for other entropy measures}\label{sec:renyi_cone}
Entropy vectors may be computed in terms of other entropy measures,
for instance in terms of the $\alpha$-R\'{e}nyi
entropies~\cite{Renyi1960_MeasOfEntrAndInf}. For a quantum state
$\rho_\mathrm{X}$, the $\alpha$-R\'{e}nyi entropy is
$H_{\alpha}(X):=\frac{1}{1-\alpha} \log \operatorname{tr}
\rho_\mathrm{X }^{\alpha},$ for $\alpha\in(0,\infty)\setminus\{1\}$,
the cases $\alpha=0,1,\infty$ are defined via the relevant limits
(note that $H_1(X)=H(X)$). \footnote{Classical $\alpha$-R\'{e}nyi entropies are
included in this definition when considering diagonal states.}

One may expect that useful constraints on the compatible distributions
can be derived from such entropy vectors.  For $0<\alpha<1$ and
$\alpha>1$ such constraints were analysed in~\cite{Linden2013a}. In
the classical case positivity and monotonicity are the only linear
constraints on the corresponding entropy vectors for any
$\alpha\neq 0,1$.  For multi-party quantum states monotonicity does
not hold for any $\alpha$, like in the case of the von Neumann
entropy. For $0<\alpha<1$, there are no constraints on the allowed
entropy vectors except for positivity, whereas for $\alpha>1$ there
are constraints, but these are non-linear.  The lack of further linear
inequalities that generally hold limits the usefulness of entropy
vectors using $\alpha$-R\'enyi entropies for analysing causal
structures. To our knowledge it is not known how or whether non-linear
inequalities for R\'{e}nyi entropies may be employed for this task.
The case $\alpha=0$, where
${H_0(X)= \log \operatorname{rank}\rho_\mathrm{X},}$, has been
considered separately in~\cite{Cadney2012a}, where it was shown that
further linear inequalities hold. However, only bi-partitions of the
parties were considered and the generalisation to full entropy vectors
is still to be explored.

The above considerations do not mention conditional entropy and hence
could be taken with the definition
${H_{\alpha}(X \mid Y)}:=H_{\alpha}(X Y)-H_{\alpha}(Y)$.  Alternatively,
one may consider a definition of the R\'{e}nyi conditional entropy,
for which
$H_\alpha(X|YZ) \leq H_\alpha(X|Y)$~\cite{Petz,TCR,MDSFT,FL,Beigi}.
With the latter definition, the conditional R\'{e}nyi entropy cannot
be expressed as a difference of unconditional entropies, and so to use
entropy vectors we would need to consider the conditional entropies as
separate components. Along these lines, one may also think about
combining R\'{e}nyi entropies for different values of $\alpha$ and to
use appropriate chain rules~\cite{Dupuis2015}. Because of the large
increase in the number of variables compared to the number of
constraints it is not clear that this will yield useful new
conditions.

\subsection{Polynomial restrictions on compatible distributions}
The probabilistic characterisation of causal structures, depends (in
general) on the dimensionality of the observed
variables. Computational hardness results suggest that a full
characterisation is unlikely to be feasible, except in small
cases~\cite{Pitowsky1991,Avis2004}.  Recent progress has been made
with the development of procedures to construct polynomial Bell
inequalities. A method that resorts to linear programming
techniques~\cite{Chaves2015a} has lead to the derivation of new
inequalities for the bilocality scenario (as well as a related
four-party scenario).  Another, iterative procedure allows for
enlarging networks by adding a party to a network in a particular
way\footnote{Here, adding a party means adding one observed input and
  one observed output node as well as an unobserved parent for the
  output, the latter may causally influence one other output random
  variable in the network.}. This allows for the constructions of
non-linear inequalities for the latter, enlarged network from
inequalities that are valid for the former~\cite{Rosset2015}.

Furthermore, a recent approach relies on considering enlarged
networks, so called inflations, and inferring causal constraints from
those~\cite{Wolfe2016, Navascues2017}. Inflated networks may contain
several copies of a variable that each have the same dependencies on
ancestors (the latter may also exist in several instances) and which
share the same distributions with their originals. Such inflations
allow for the derivation of probabilistic inequalities that restrict
the set of compatible distributions.  These ideas also bear some
resemblance to the procedures in~\cite{Kela2017}, in the sense that
they employ the idea that certain marginal distributions may be
obtained from different networks; they are, however, much more focused
on causal structures featuring interesting independence constraints.
Inflations allowed the authors of~\cite{Wolfe2016} to refute certain
distributions as incompatible with the triangle causal structure from
Figure~\ref{fig:instrumental}(c), in particular the so called
W-distribution which could neither be proven to be incompatible
entropically nor with the covariance matrix approach below.

\subsection{Semidefinite tests relying on covariance matrices}
One may look for mappings of the distribution of a set of observed
variables that encode causal structure beyond considering
entropies. For causal structures with two generations, i.e., one
generation of unobserved variables as ancestors of one generation of
observed nodes, a technique has been found using covariance
matrices~\cite{Kela2017}.  Each observed variable is mapped to a
vector-valued random variable and the covariance matrix of the direct
sum of these variables is considered. Due to the law of total
expectation, this matrix allows for a certain decomposition depending
on the causal structure.  For a particular observed distribution and
its covariance matrix, the existence of such a decomposition may be
tested via semidefinite programming. The relation of this technique to
the entropy vector method is not yet well understood. A partial
analysis considering several examples is given in Section X of
\cite{Kela2017}.

\section{Open Problems} \label{sec:conclusion} The entropy vector
approach has led to many certificates for the incompatibility of
correlations with causal structures. However, we are still
lacking a general understanding of how well entropic relations can
approximate the set of achievable correlations. Firstly, the
non-injective mapping from probabilities to entropies is not
sufficiently understood and secondly, the current methods employ
further approximations, e.g.\ by restricting the number of non-Shannon
inequalities that can be considered at a time. It is as yet unknown
whether the entropy vector method (without post-selection) can ever
distinguish correlations that arise from classical, quantum and more
general non-signalling resources. Such insights may also inform the
question of whether there exist novel inequalities for the von Neumann
entropy of multi-party quantum states.

The post-selection technique allows for the derivation of additional
constraints that may distinguish quantum from classically achievable
correlations in the Bell scenario and possibly in other
examples. However, the method relies on the causal structure featuring
parentless observed nodes, hence it is not always applicable (see
e.g.\ the triangle scenario).  In such situations, one may try to
combine the entropic techniques reviewed here with the inflation
method~\cite{Wolfe2016}, which might allow for further entropic
analysis of several causal structures, e.g., of the triangle scenario.

Criteria to certify whether a set of entropic constraints is able to
detect non-classical correlations are currently not available. For
many of the established entropic constraints on classical causal
structures it is unknown whether or not they are also valid for the
corresponding quantum structure.  In the case of the Bell scenario
this problem has been overcome. It has been shown that the known
entropic constraints are even sufficient for detecting any
non-classical correlations~\cite{Chaves2013}. However, since the proof
is specific to the scenario, finding a systematic tool to analyse the
scope of the entropic techniques remains open.

\acknowledgments
We thank Rafael Chaves and Costantino Budroni for confirming details of~\cite{Chaves2016}.
RC is supported by the EPSRC's Quantum Communications Hub
(grant no.\ EP/M013472/1) and by an EPSRC First Grant (grant no.\
EP/P016588/1).


\appendix

\section{Entropy inequalities for Example~10}

In the following we provide the basic inequalities for the quantum instrumental scenario, $IC^{\rm Q}$, i.e., the the constraints making up the matrix $M_B(IC^{\rm Q})$. 
\begin{eqnarray}
I(A_Y:A_Z)&\geq&0 \\
I(A_Y:X)&\geq&0 \\ 
I(A_Z:X)&\geq&0 \\ 
I(A_Y:A_Z|X)&\geq&0 \\ 
I(A_Y:X|A_Z)&\geq&0 \\ 
I(A_Z:X|A_Y)&\geq&0 \\ 
I(A_Y:Z)&\geq&0 \\ 
I(X:Z)&\geq&0 \\ 
I(X:Z|A_Y)&\geq&0 \\ 
I(A_Y:Z|X)&\geq&0 \\ 
I(A_Y:X|Z)&\geq&0 \\ 
I(X:Y)&\geq&0 \\ 
I(Y:Z)&\geq&0 \\ 
I(Y:Z|X)&\geq&0 \\ 
I(X:Z|Y)&\geq&0 \\ 
I(X:Y|Z)&\geq&0 \\ 
H(A_Z|X)&\geq&0 \\ 
H(A_Y A_Z|X)&\geq&0 \\ 
H(X|A_Y A_Z)&\geq&0 \\
H(A_Y|XZ)&\geq&0 \\ 
H(X|A_YZ)&\geq&0 \\ 
H(Z|A_YX)&\geq&0 \\
H(X|YZ)&\geq&0 \\
H(Y|XZ)&\geq&0 \\ 
H(Z|XY)&\geq&0 \\ 
H(A_Z|A_Y)+H(A_Z|X) &\geq&0 \\ 
H(A_Y|A_Z)+H(A_Y|X) &\geq&0 \\ 
H(A_Z|A_YX)+H(A_Z) &\geq&0 \\ 
H(A_Y|A_ZX)+H(A_Y) &\geq&0 
\end{eqnarray}
Independence constraints and data processing inequalities are provided
in the main text.  If we include these and remove redundant
inequalities we obtain the following set of constraints, which for
convenience we give in matrix form (such that $\Gamma(IC^{{\rm Q}})=\{v\in\mathbb{R}^{15}_{\geq
  0}\mid M\cdot v\geq 0\}$):
  
$$M=\left(
\begin{array}{ccccccccccccccc}
 0 & 0 & 0 & 0 & 0 & 0 & 0 & 1 & 0 & 0 & 0 & -1 & 0 & -1 & 1 \\
 0 & 0 & 0 & 0 & 0 & 0 & 0 & 0 & 1 & 0 & -1 & 0 & -1 & 1 & 0 \\
 0 & 0 & -1 & 0 & 0 & -1 & 0 & 0 & 0 & 0 & 0 & 0 & 1 & 0 & 0 \\
 0 & 0 & 0 & 0 & -1 & 0 & 0 & 0 & 0 & 0 & 1 & 1 & 0 & 0 & -1 \\
 0 & 0 & 0 & -1 & 0 & 0 & 0 & 0 & 0 & 1 & 0 & 1 & 0 & 0 & -1 \\
 0 & 0 & -1 & 0 & 0 & 0 & 0 & 0 & 0 & 1 & 1 & 0 & 0 & 0 & -1 \\
 0 & 0 & 0 & 1 & 1 & 0 & 0 & 0 & 0 & 0 & 0 & -1 & 0 & 0 & 0 \\
 0 & 0 & 1 & 0 & 1 & 0 & 0 & 0 & 0 & 0 & -1 & 0 & 0 & 0 & 0 \\
 0 & 0 & 1 & 1 & 0 & 0 & 0 & 0 & 0 & -1 & 0 & 0 & 0 & 0 & 0 \\
 1 & 0 & 0 & 0 & 1 & 0 & 0 & -1 & 0 & 0 & 0 & 0 & 0 & 0 & 0 \\
 1 & 0 & 1 & 0 & 0 & 0 & -1 & 0 & 0 & 0 & 0 & 0 & 0 & 0 & 0 \\
 0 & -1 & 0 & 0 & 0 & 1 & 0 & 0 & 1 & 0 & 0 & 0 & -1 & 0 & 0 \\
 -1 & 0 & 0 & 0 & 0 & 1 & 1 & 0 & 0 & 0 & 0 & 0 & -1 & 0 & 0 \\
 0 & 1 & 1 & 0 & 0 & 0 & 0 & 0 & -1 & 0 & 0 & 0 & 0 & 0 & 0 \\
 0 & 0 & 0 & 0 & 0 & 0 & 0 & 0 & 0 & -1 & 0 & 0 & 0 & 0 & 1 \\
 0 & 0 & 0 & 0 & 0 & 0 & 0 & 0 & 0 & 0 & -1 & 0 & 0 & 0 & 1 \\
 0 & 0 & 0 & 0 & 0 & 0 & -1 & 0 & 0 & 0 & 0 & 0 & 0 & 1 & 0 \\
 0 & 0 & 0 & 0 & 0 & 0 & 0 & -1 & 0 & 0 & 0 & 0 & 0 & 1 & 0 \\
 0 & 0 & 0 & 0 & 0 & 0 & 0 & 0 & 0 & 0 & -1 & 0 & 0 & 1 & 0 \\
 1 & 0 & 0 & 0 & 0 & 0 & 0 & 0 & -1 & 0 & 0 & 0 & 1 & 0 & 0 \\
 0 & 1 & 0 & 0 & 0 & 0 & -1 & 0 & 0 & 0 & 0 & 0 & 1 & 0 & 0 
\end{array}
\right)$$

\section{Entropy inequalities for Example~16}
The causal structure of Figure~\ref{fig:structure15}, previously considered in~\cite{Henson2014, Pienaar2016} is analysed entropically by means of the entropic post-selection technique.
\begin{figure}
\centering
\resizebox{0.7\columnwidth}{!}{%
\begin{tikzpicture}[scale=0.72]
\node (0) at (-2.75,3) {$(a)$};
\node[draw=black,circle,scale=0.75] (X) at (-2,2) {$X$};
\node[draw=black,circle,scale=0.75] (Y) at (2,2) {$Y$};
\node[draw=black,circle,scale=0.75] (Z) at (0,-1.46) {$Z$};
\node (A) at (1,0.28) {$A$};
\node (B) at (-1,0.28) {$B$};
\node[draw=black,circle,scale=0.75] (C) at (0,2) {$C$};

\draw [->,>=stealth] (A)--(Y); 
\draw [->,>=stealth] (A)--(Z); 
\draw [->,>=stealth] (B)--(X); 
\draw [->,>=stealth] (B)--(Z); 
\draw [->,>=stealth] (C)--(X); 
\draw [->,>=stealth] (C)--(Y); 

\node (X0) at (5.25,3) {$(b)$};
\node[draw=black,circle,scale=0.75] (XX0) at (6,2) {$X0$};
\node[draw=black,circle,scale=0.75] (XX1) at (7,2) {$X1$};
\node[draw=black,circle,scale=0.75] (XY0) at (9,2) {$Y0$};
\node[draw=black,circle,scale=0.75] (XY1) at (10,2) {$Y1$};
\node[draw=black,circle,scale=0.75] (XZ) at (8,-1.5) {$Z$};
\node (XA) at (9,0.5) {$A$};
\node (XB) at (7,0.5) {$B$};

\draw [->,>=stealth] (XA)--(XY0); 
\draw [->,>=stealth] (XA)--(XY1);
\draw [->,>=stealth] (XA)--(XZ); 
\draw [->,>=stealth] (XB)--(XX0); 
\draw [->,>=stealth] (XB)--(XX1); 
\draw [->,>=stealth] (XB)--(XZ); 

\end{tikzpicture}
}%
\caption{Post-selecting (a) on a binary observed variable $C$ leads to the causal structure (b).}
\label{fig:structure15}
\end{figure}
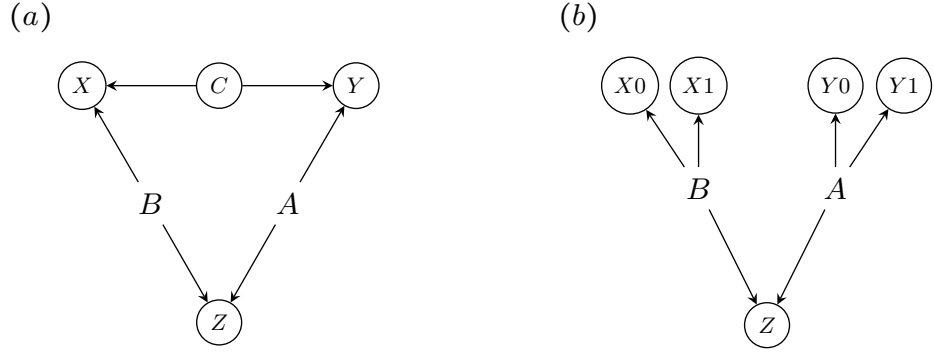
The outer approximation to the entropic cone of the causal structure of Figure~\ref{fig:structure15} is computed and marginalised to vectors \begin{multline*}\left(H(Z),~H(X_{\rm \mid C=0}),~H(X_{\rm \mid C=1}),~H(Y_{\rm \mid C=0}),~H(Y_{\rm \mid C=1}),~H(X_{\rm \mid C=0}Z),~H(X_{\rm \mid C=1}Z),~H(Y_{\rm \mid C=0}Z), \right. \\ \left.~H(Y_{\rm \mid C=1}Z),~H(X_{\rm \mid C=0}Y_{\rm \mid C=0}),~H(X_{\rm \mid C=1}Y_{\rm \mid C=1}),~H(X_{\rm \mid C=0}Y_{\rm \mid C=0}Z),~H(X_{\rm \mid C=1}Y_{\rm \mid C=1}Z) \right).\end{multline*} 
The following $14$ extremal rays are obtained from this computation, where each ray is represented by one particular vector on it. The tip of this pointed polyhedral cone is the zero-vector. 

$$\arraycolsep=1.2pt 
\begin{array}{cccccccccccccccc}
(  1) &&& 0 & 0 & 0 & 0 & 1 & 0 & 0 & 0 & 1 & 0 & 1 & 0 & 1  \\
(  2) &&& 0 & 0 & 0 & 1 & 0 & 0 & 0 & 1 & 0 & 1 & 0 & 1 & 0  \\
(  3) &&& 0 & 0 & 1 & 0 & 0 & 0 & 1 & 0 & 0 & 0 & 1 & 0 & 1  \\
(  4) &&& 0 & 1 & 0 & 0 & 0 & 1 & 0 & 0 & 0 & 1 & 0 & 1 & 0  \\
(  5) &&& 1 & 1 & 1 & 1 & 1 & 2 & 2 & 2 & 2 & 2 & 2 & 2 & 2  \\
(  6) &&& 1 & 0 & 1 & 0 & 1 & 1 & 2 & 1 & 2 & 0 & 2 & 1 & 2  \\
(  7) &&& 1 & 1 & 0 & 1 & 0 & 2 & 1 & 2 & 1 & 2 & 0 & 2 & 1  \\
(  8) &&& 1 & 0 & 0 & 0 & 0 & 1 & 1 & 1 & 1 & 0 & 0 & 1 & 1  \\
(  9) &&& 1 & 0 & 0 & 0 & 1 & 1 & 1 & 1 & 1 & 0 & 1 & 1 & 1  \\
( 10) &&& 1 & 0 & 0 & 1 & 0 & 1 & 1 & 1 & 1 & 1 & 0 & 1 & 1  \\
( 11) &&& 1 & 0 & 1 & 0 & 0 & 1 & 1 & 1 & 1 & 0 & 1 & 1 & 1  \\
( 12) &&& 1 & 1 & 0 & 0 & 0 & 1 & 1 & 1 & 1 & 1 & 0 & 1 & 1  \\
( 13) &&& 1 & 0 & 0 & 1 & 1 & 1 & 1 & 1 & 1 & 1 & 1 & 1 & 1  \\
( 14) &&& 1 & 1 & 1 & 0 & 0 & 1 & 1 & 1 & 1 & 1 & 1 & 1 & 1  
\end{array}
$$

\setcounter{equation}{0}

The corresponding inequality description is given by the $2$
equalities and $18$ inequalities (or equivalently $22$ inequalities if
each equality is written as two inequalities).
\begin{eqnarray}
H(X_{\rm \mid C=0})+ H(Y_{\rm \mid C=0})&=&H(X_{\rm \mid C=0}Y_{\rm \mid C=0})\\
H(X_{\rm \mid C=1})+ H(Y_{\rm \mid C=1})&=&H(X_{\rm \mid C=1}Y_{\rm \mid C=1})\\
H(X_{\rm \mid C=1}Y_{\rm \mid C=1})&\leq&H(X_{\rm \mid C=1}Y_{\rm \mid C=1} Z)\\
H(X_{\rm \mid C=0}Y_{\rm \mid C=0})&\leq&H(X_{\rm \mid C=0}Y_{\rm \mid C=0} Z)\\
H(Y_{\rm \mid C=1}Z)&\leq&H(X_{\rm \mid C=1}Y_{\rm \mid C=1} Z)\\
H(Y_{\rm \mid C=0}Z)&\leq&H(X_{\rm \mid C=0}Y_{\rm \mid C=0} Z)\\
H(X_{\rm \mid C=1}Z)&\leq&H(X_{\rm \mid C=1}Y_{\rm \mid C=1} Z)\\
H(X_{\rm \mid C=0}Z)&\leq&H(X_{\rm \mid C=0}Y_{\rm \mid C=0} Z)\\
H(Y_{\rm \mid C=0}Z)&\leq&H(Z) + H(Y_{\rm \mid C=0})\\
H(Y_{\rm \mid C=1}Z)&\leq&H(Z) + H(Y_{\rm \mid C=1})\\
H(Y_{\rm \mid C=1}) + H(X_{\rm \mid C=1}Z) &\leq&H(Z) + H(X_{\rm \mid C=1}Y_{\rm \mid C=1})\\
H(Y_{\rm \mid C=0}) + H(X_{\rm \mid C=0}Z) &\leq&H(Z) + H(X_{\rm \mid C=0}Y_{\rm \mid C=0})\\
H(Z) + H(X_{\rm \mid C=0}Y_{\rm \mid C=0} Z)&\leq& H(X_{\rm \mid C=0}Z)+ H(Y_{\rm \mid C=0}Z) \\
H(Z) + H(X_{\rm \mid C=1}Y_{\rm \mid C=1} Z)&\leq& H(X_{\rm \mid C=1}Z)+ H(Y_{\rm \mid C=1}Z) \\
H(Y_{\rm \mid C=1}) + H(X_{\rm \mid C=1}Z)+ H(X_{\rm \mid C=0}Y_{\rm \mid C=0}) &\leq&H(Y_{\rm \mid C=0}) + H(X_{\rm \mid C=0}Z)+H(X_{\rm \mid C=1}Y_{\rm \mid C=1}Z)\\
H(Y_{\rm \mid C=1}) + H(Y_{\rm \mid C=0}Z)+ H(X_{\rm \mid C=0}Y_{\rm \mid C=0}) &\leq&H(Y_{\rm \mid C=0}) + H(Y_{\rm \mid C=1}Z)+H(X_{\rm \mid C=0}Y_{\rm \mid C=0}Z)\\
H(X_{\rm \mid C=1} Z)+H(Y_{\rm \mid C=1} Z)+H(X_{\rm \mid C=0} Y_{\rm \mid C=0})&\leq&H(X_{\rm \mid C=0} Z)+H(Y_{\rm \mid C=0} Z)+H(X_{\rm \mid C=1}Y_{\rm \mid C=1}Z)\\
H(X_{\rm \mid C=0} Z)+H(Y_{\rm \mid C=0} Z)+H(X_{\rm \mid C=1} Y_{\rm \mid C=1})&\leq&H(X_{\rm \mid C=1} Z)+H(Y_{\rm \mid C=1} Z) +H(X_{\rm \mid C=0}Y_{\rm \mid C=0}Z)\\
H(Y_{\rm \mid C=0})+H(Y_{\rm \mid C=1} Z)+H(X_{\rm \mid C=1}Y_{\rm \mid C=1})&\leq&H(Y_{\rm \mid C=1})+H(Y_{\rm \mid C=0} Z)+H(X_{\rm \mid C=1}Y_{\rm \mid C=1}Z) \\
H(Y_{\rm \mid C=0})+H(X_{\rm \mid C=0} Z)+H(X_{\rm \mid C=1}Y_{\rm \mid C=1})&\leq&H(Y_{\rm \mid C=1})+H(X_{\rm \mid C=1} Z)+H(X_{\rm \mid C=0}Y_{\rm \mid C=0}Z). 
\end{eqnarray}

\end{document}